\documentclass[
  aps,prl,reprint,superscriptaddress,nofootinbib,longbibliography,floatfix
]{revtex4-2}

\usepackage[T1]{fontenc}
\usepackage[english]{babel}
\usepackage{amsmath,amssymb,bm}
\usepackage{graphicx,microtype}
\usepackage[hidelinks]{hyperref}
\newtheorem{theorem}{Theorem}
\newtheorem{lemma}[theorem]{Lemma}
\newtheorem{corollary}[theorem]{Corollary}

\newcommand{\F}{\mathbb F}
\newcommand{\AME}{\operatorname{AME}}
\newcommand{\Norm}{\operatorname{Norm}}
\newcommand{\Tr}{\operatorname{Tr}}
\newcommand{\tr}{\operatorname{tr}}
\newcommand{\ket}[1]{\lvert #1\rangle}
\newcommand{\bra}[1]{\langle #1\rvert}

\begin{document}

\title{\texorpdfstring{Absolutely Maximally Entangled States\\
for Odd Prime Power Local Dimension}{Absolutely Maximally Entangled States
for Odd Prime Power Local Dimension}}

\author{Mykhailo Hontarenko}
\email{mykhailo.hontarenko@uj.edu.pl}
\affiliation{Institute of Theoretical Physics, Faculty of Physics, Astronomy and Applied Computer Science,
Jagiellonian University, ul.~\L{}ojasiewicza 11, 30-348 Krak\'ow, Poland}
\affiliation{Doctoral School of Exact and Natural Sciences, Jagiellonian University,
ul.~\L{}ojasiewicza 11, 30-348 Krak\'ow, Poland}

\author{Karol \.{Z}yczkowski}
\email{karol.zyczkowski@uj.edu.pl}
\affiliation{Institute of Theoretical Physics, Faculty of Physics, Astronomy and Applied Computer Science,
Jagiellonian University, ul.~\L{}ojasiewicza 11, 30-348 Krak\'ow, Poland}
\affiliation{Center for Theoretical Physics, Polish Academy of Sciences,
al.~Lotnik\'ow 32/46, 02-668 Warszawa, Poland}

\date{13 November 2026}

\begin{abstract}
An absolutely maximally entangled (AME) state is a pure state of $N$
parties with $d$ levels each whose reduced states on every subset of
at most half the parties are maximally mixed. Such states provide perfect tensors and optimal
quantum error-correcting codes. Determining the values of $N$ and $d$
for which they exist is one of the central problems in multipartite entanglement.
For every odd prime power, $q=p^e\ge3$, we construct a stabilizer $\AME(2q,q)$ state and prove the existence of stabilizer $\AME(N,q)$ states for $N=2,3,\ldots,2q$.
The construction uses a closed-form $q\times q$ bordered-circulant
matrix $A_q$ over the finite field $\F_{q^2}$ to generate a Hermitian self-dual maximum
distance separable (MDS) code $[2q,q,q+1]_{q^2}$.
These codes lie outside the generalized Reed--Solomon family,
including its extended forms.
In suitable local bases, the parent states yield Butson Hadamard
matrices of order $q^q$ with $p$th-root entries, which become
$q$-unitary after normalization.
Paired projections on antipodal coordinates give all even lengths from
$q+1$ to $2q$; one-party projections and standard short-length
constructions complete the range. Weighted submatrices give further
explicit realizations.
\end{abstract}

\maketitle

The Schmidt coefficients characterize the entanglement of a
bipartite pure state up to local unitary transformations
\cite{Nielsen1999}. For three or more parties, entanglement must be
considered across the different bipartitions of the system
\cite{Review2026}. A state that is maximally entangled across every
bipartition is called absolutely maximally entangled
\cite{Helwig2012}.

A pure state $\ket{\Psi}\in(\mathbb C^d)^{\otimes N}$ is denoted
$\AME(N,d)$ when its reduced density matrix on every subset $S$ of
at most half the parties satisfies
\begin{equation}
\rho_S=\operatorname{tr}_{S^c}\ket{\Psi}\!\bra{\Psi}
 =\frac{I_S}{d^{|S|}},
 \qquad |S|\le\lfloor N/2\rfloor .
 \label{eq:ame}
\end{equation}
Here $S^c$ is the complement of $S$, $I_S$ is the identity on the
parties in $S$, and $|S|$ is their number. Every measurement on this
subsystem therefore has the statistics of the maximally mixed state.
For $N=2m$, the amplitude tensor is proportional to an isometry from
any $m$ inputs to the remaining $m$ outputs, giving a perfect tensor.  AME states connect multipartite entanglement with optimal quantum
error correction, quantum secret sharing, parallel teleportation,
multiunitary matrices, and holographic tensor networks
\cite{Scott2004,Helwig2012,Goyeneche2015,Pastawski2015,Review2026}.

For a prime-power $d$, Reed--Solomon codes give minimal-support states in the
standard range $N\le d+1$ \cite{Raissi2018}, where  minimal support means $d^{\lfloor N/2\rfloor}$ nonzero amplitudes
in a product basis. Quantum-code constructions
provide further families and connect the problem with finite geometry
\cite{GrasslRotteler2015,BallCentellesHuber2023}, while many parameter pairs
remain unresolved \cite{AMETable}.  The seven-party case is now completely
settled: an $\AME(7,d)$ state exists \cite{Shi2026} if and only if $d\ge3$.
Recent AI-assisted searches produced five isolated parameter pairs
from three explicit Hermitian self-dual maximum distance separable (MDS)
codes, using exact minor checks and one-party projections
\cite{BevinsBidav2026}.
An infinite family obtained by this approach requires a systematic
generator $G=[I\mid A]$ with two properties: its code must be Hermitian
self-dual, and every square submatrix of $A$ must be nonsingular.
The latter property is called superregularity, or total nonsingularity
\cite{Keri2006}. We give an explicit matrix and prove both properties for
every odd prime power $q$, obtaining $\AME(2q,q)$. Special paired
projections then give the even lengths $q+1,q+3,\ldots,2q$; a further
one-party projection gives each odd partner.

Together with the standard short-length construction, this gives
\begin{equation}
 \AME(N,q)\quad\text{for every}\quad 2\le N\le2q,
 \label{eq:families}
\end{equation}
where $q=p^e$, $p$ is an odd prime, and $e\ge1$.
Our result also includes nonprime local dimensions such as $9$, $25$,
and $27$. Table~\ref{tab:finite-atlas} lists nine explicit even-party
parents and their odd-party projections, all covered by
Eq.~\eqref{eq:families}. These matrices provide reproducible
representatives. We identify $M_{34,17}=A_{17}$ explicitly;
equivalence of the remaining matrices to the general representatives
is a separate question.
The Supplemental Material (SM) gives the finite-field conventions,
matrix entries, proofs, and exact minor checks \cite{SM}.
Figure~\ref{fig:ame-map} places the constructions in an AME existence
map whose historical sources and stabilizer conventions are specified
in the SM. The map also includes explicit composite-dimension
constructions from the repositories in
Refs.~\cite{KitamuraAME96,AtlasAME910,KitamuraAME1110}.
Probability representations and Hadamard-block structures
of the non-stabilizer states $\AME(7,6)$ and $\AME(9,6)$
are discussed in our companion paper
\cite{HontarenkoZyczkowskiLatin2026}.
\begin{figure*}[tp]
 \centering
 \includegraphics[width=\textwidth]{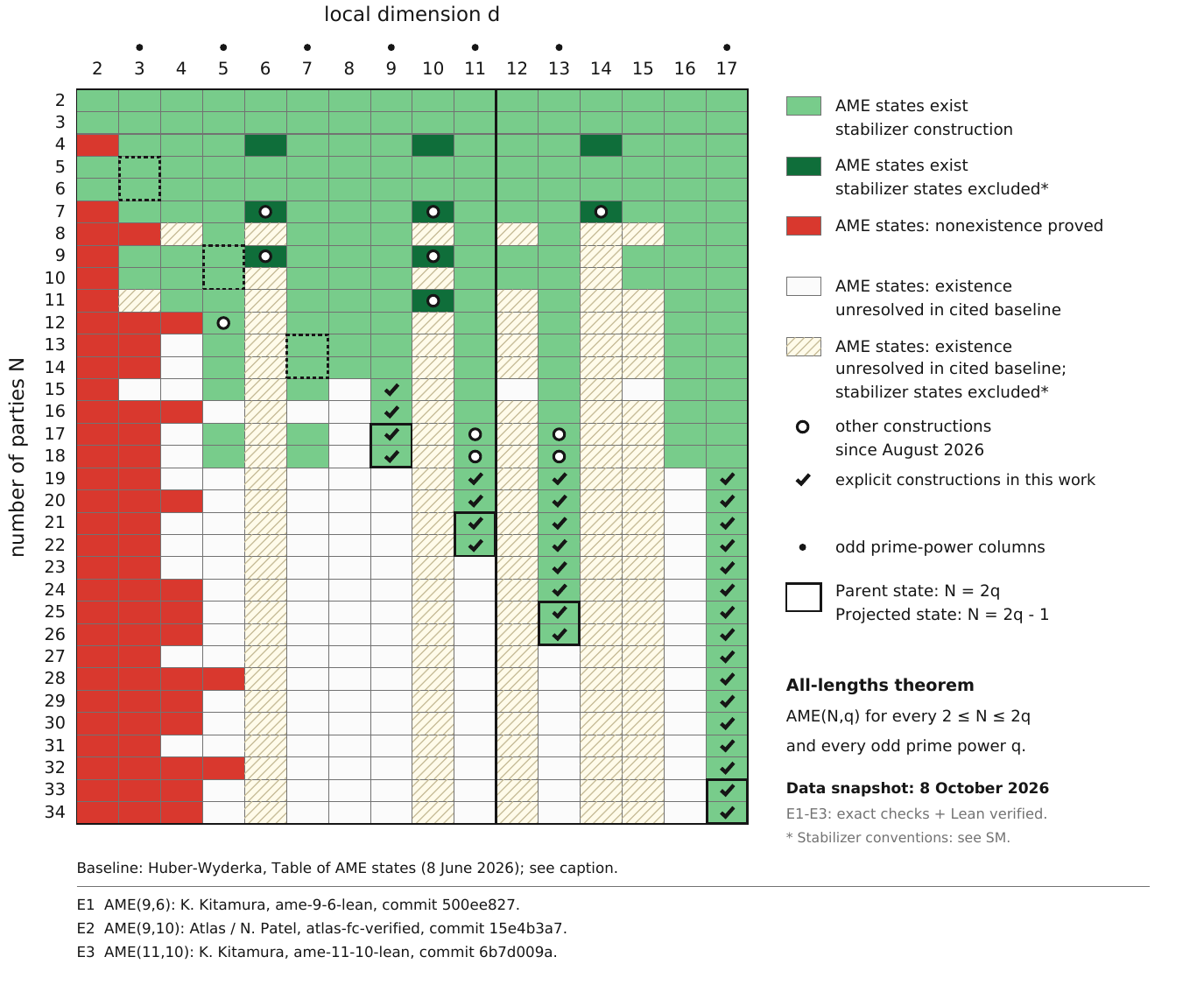}
 \caption{Existence map of AME states for $2\le d\le17$ and $2\le N\le34$.
 Dots above the columns mark odd prime powers $q$; all cells with
 $2\le N\le2q$ in these columns follow from Theorem~\ref{thm:main}.
 Black frames enclose the parent state at $N=2q$ and its one-party
 projection at $N=2q-1$ (dashed for previously known $q\le7$).
 Black checks identify explicit examples
 in this work; open circles mark other constructions since August 2026.
 Background cells use the archived
 Huber--Wyderka table and literature through 5 September 2026
\cite{AMETable,Raissi2018,Review2026,GrasslRotteler2015,
 Shi2026,BevinsBidav2026}, supplemented by the present theorem and
 the $\AME(9,6)$, $\AME(9,10)$ and $\AME(11,10)$ constructions
 \cite{KitamuraAME96,AtlasAME910,KitamuraAME1110}.
 Light green denotes stabilizer constructions. Dark green denotes
 known AME states for which stabilizer states are excluded in the
 convention specified in the SM. Hatching marks the same stabilizer
 exclusion where AME existence is unsettled in the cited sources
 \cite{Cha2026,Wojcik2026,BallMorenoSimoens2025,Danielsen2012}.
 Red denotes proven nonexistence; white denotes a case unsettled in
 those sources. Stabilizer conventions, existence sources, and a larger
 map including $q=19,23,25,27,29$ are given in the SM \cite{SM}.}
 \label{fig:ame-map}
\end{figure*}

Bordered-circulant self-dual MDS matrices have previously served as a
starting point for computational searches
\cite{GrasslGulliver2008,GrasslGulliver2009}, and
Cauchy and bordered-Cauchy matrices are classical routes to MDS codes
\cite{RothSeroussi1985,RothLempel1989}.
Our matrix $A_q$ retains the bordered-circulant structure, with
entries given by a rational formula. Its finite-field Fourier
coefficients establish Hermitian orthogonality; a separate root-counting
argument proves that every square minor is nonzero.

\smallskip
\noindent\textit{The coding criterion.---}
The Hermitian stabilizer construction relates the AME condition to
the self-duality and distance of a classical code \cite{Ketkar2006}.
We work over the finite field $\F_{q^2}$ with $q^2$ elements and write
\begin{equation}
 \bar z=z^q,\qquad \Norm(z)=z\bar z\in\F_q .
 \label{eq:norm}
\end{equation}
Here $\operatorname{Norm}$ denotes the field norm; $N$ remains the number
of parties.
On matrices the bar acts entrywise, so
$\overline{A_q}^{\,\mathsf T}$ is the Hermitian transpose over
$\F_{q^2}/\F_q$.
Let $I_q$ denote the $q\times q$ identity matrix.
It is enough to find a $q\times q$ matrix $A_q$ over $\F_{q^2}$ such that
\begin{align}
 A_q\overline{A_q}^{\,\mathsf T}&=-I_q,\nonumber\\
 \det A_q[R,C]&\ne0
 \quad\bigl(0<|R|=|C|\le q\bigr).
 \label{eq:matrix-goal}
\end{align}
Here $A_q[R,C]$ denotes the submatrix with row set $R$ and column
set $C$. The determinant condition holds for all square submatrices
and expresses superregularity \cite{Keri2006}. For a systematic
generator, this is precisely the MDS criterion
\cite[Theorem~2.4.3]{HuffmanPless2003}.
Adjoining the identity gives the $q\times2q$ generator matrix
\begin{equation}
 G_q=[\,I_q\mid A_q\,],\qquad
 \mathcal C_q=\{\bm uG_q:\bm u\in\F_{q^2}^{\,q}\}
 \subseteq\F_{q^2}^{\,2q}.
 \label{eq:classical-code}
\end{equation}
The first condition in Eq.~\eqref{eq:matrix-goal} makes $\mathcal C_q$
Hermitian self-orthogonal.  Since $G_q$ has rank $q$, the code has half the
ambient dimension and is therefore Hermitian self-dual.
For completeness, choose any $q$
columns of $G_q$: $q-s$ from $I_q$, with labels $J$, and $s$ from $A_q$,
with labels $C$.  Expanding the determinant along the identity columns
leaves $\pm\det A_q[R,C]$, where $R$ consists of the $s$ rows outside
$J$.  This is nonzero by superregularity when $s>0$.
If $s=0$, all columns come from the identity and the determinant is $\pm1$.  Thus any $q$ columns of $G_q$ are independent.  A nonzero codeword $\bm uG_q$ cannot have $q$ zero coordinates: those coordinates would give an invertible linear
system forcing $\bm u=0$.  Its Hamming weight is therefore at least $q+1$.  The classical Singleton bound \cite{Singleton1964} gives the reverse inequality, so the distance is exactly $q+1$.

We write $[n,k,d]_{q^2}$ for a classical linear code over
$\F_{q^2}$ of length $n$, dimension $k$, and minimum Hamming
distance $d$. The quantum notation $[[n,k,d]]_q$ describes
a code on $n$ $q$-level systems, with code-space dimension
$q^k$ and distance $d$; for $k=0$, this space is one-dimensional
and specifies a single quantum state.
Our construction follows the chain
\begin{equation*}
 A_q
 \longrightarrow [2q,q,q+1]_{q^2}
 \longrightarrow [[2q,0,q+1]]_q
 \longrightarrow \AME(2q,q).
\end{equation*}

\begin{table}[tbp]
\caption{Explicit parent states and their normalized one-party projections.
An asterisk ($*$) marks a parent with $N=2q$; the three parents with
$N=q+3$ are also covered by the supplementary construction.}
\label{tab:finite-atlas}
\begin{ruledtabular}
\begin{tabular}{ccc}
$q$ & Parent state & Projected state \\
\colrule
9  & $\AME(16,9)$      & $\AME(15,9)$ \\
9  & $\AME(18,9)^{*}$  & $\AME(17,9)$ \\
11 & $\AME(22,11)^{*}$ & $\AME(21,11)$ \\
13 & $\AME(20,13)$     & $\AME(19,13)$ \\
13 & $\AME(26,13)^{*}$ & $\AME(25,13)$ \\
17 & $\AME(20,17)$     & $\AME(19,17)$ \\
17 & $\AME(34,17)^{*}$ & $\AME(33,17)$ \\
25 & $\AME(28,25)$     & $\AME(27,25)$ \\
29 & $\AME(32,29)$     & $\AME(31,29)$ \\
\end{tabular}
\end{ruledtabular}
\end{table}

The classical parameters are length $2q$, dimension $q$ over
$\F_{q^2}$, and minimum distance $q+1$. The Hermitian stabilizer
construction gives the pure quantum code $[[2q,0,q+1]]_q$
\cite{Rains1999,Ketkar2006}. In the pure-state convention, its distance
is the minimum nonzero Hamming weight of the self-dual classical code,
namely $q+1$. Therefore it  saturates the quantum Singleton bound
\cite{Singleton1964,GrasslHuberWinter2022}.
To compare these codes with generalized Reed--Solomon (GRS) codes,
we use the length restrictions imposed by Hermitian self-orthogonality.
Grassl and R\"otteler constructed quantum MDS (QMDS) codes,
including codes of length $q^2+1$ \cite{GrasslRotteler2015};
Ball later identified their underlying cyclic
and constacyclic codes as GRS codes \cite{Ball2023GRS}.  Ball and Vilar
proved the length restriction conjectured in Ref.~\cite{GrasslRotteler2015}.
Theorems~2.1 and 3.1 of Ref.~\cite{BallVilar2022}, specialized to
dimension $q$, imply that a Hermitian self-orthogonal GRS
code of dimension $q$ over $\F_{q^2}$ must have length $q^2+1$; their
definition includes extended GRS codes. Since $2q<q^2+1$, our ``parent'' codes are neither GRS nor extended GRS
codes. To the best of our knowledge,
the standard QMDS families and the known classical constructions of
Hermitian self-dual MDS codes contain no uniform family with parameters
$[2q,q,q+1]_{q^2}$ for every odd prime power $q$
\cite{HuberGrassl2020,KimLee2004,GulliverKimLee2008,BallCentellesHuber2023,
ZhuWan2025,ZhaoMa2026}.

\smallskip
\noindent\textit{Explicit construction.---}
Fix an odd prime power $q$. Choose a primitive element
$g\in\F_{q^2}^\times$ and put
$\gamma=g^{q+1}$.  Then $\gamma$ generates $\F_q^\times$, is a nonsquare in $\F_q$, and
$\Norm(g)=\gamma$.  Choose also $b\in\F_{q^2}^\times$ with $\Norm(b)=2$; it
exists because the norm takes every nonzero value in $\F_q$.
A primitive element gives a convenient parametrization of the
construction. The SM develops the more general admissible choices and
their border scalings \cite{SM}.

We index the rows of $A_q$ by $\{0\}\cup\F_q^\times$ and the columns by
$\{\infty\}\cup\F_q^\times$, where $\infty$ is a formal label for the
border column.  We write $[A_q]_{x,y}$ for a single matrix entry.
For $x,y\in\F_q^\times$, define
\begin{equation}
 \begin{aligned}\relax
 [A_q]_{0,\infty}&=-b/{\bar b},&
 [A_q]_{0,y}&=[A_q]_{x,\infty}=b,\\
 [A_q]_{x,y}&=\frac{2y(y+gx)}{\gamma x^2-y^2}.
 \end{aligned}
 \label{eq:Aq-explicit}
\end{equation}
For $x,y\in\F_q^\times$, the denominator is nonzero because
$\gamma$ is a nonsquare in $\F_q$, whereas $(y/x)^2$ is a square.
The numerator is also nonzero: the primitive element
$g\in\F_{q^2}^{\times}$ lies outside $\F_q$, so it cannot equal
$-y/x\in\F_q^\times$.

For example, let $q=3$ and $\F_9=\F_3[\theta]$.
Here $\gamma=-1$, and we choose $\theta$ with $\theta^2=\gamma$,
so $\theta^2=-1$ in $\F_9$. All entries of the matrix below are
finite-field elements.
Taking $g=b=1+\theta$ gives
\begin{equation}
 A_3=
 \begin{pmatrix}
  -\theta&1+\theta&1+\theta\\
  1+\theta&1-\theta&\theta\\
  1+\theta&\theta&1-\theta
 \end{pmatrix},
 \label{eq:A3}
\end{equation}
in the row order $(0,1,-1)$ and column order $(\infty,1,-1)$.  Thus
$[\,I_3\mid A_3\,]$ gives the explicit chain
$[6,3,4]_9\to[[6,0,4]]_3\to\AME(6,3)$.
Section~\ref{supp:sec-qutrit-example} of the SM works through this
example and gives the state amplitudes explicitly.

\begin{theorem}\label{thm:main}
For every odd prime power $q=p^e$ and every integer $2\le N\le2q$,
a stabilizer $\AME(N,q)$ state exists.
In particular, the matrix in Eq.~\eqref{eq:Aq-explicit}, with $g$
generating $\F_{q^2}^{\times}$ and $\Norm(b)=2$, satisfies
Eq.~\eqref{eq:matrix-goal} and gives a stabilizer $\AME(2q,q)$ state
through the systematic generator of the form $[I_q\mid A_q]$.
Paired projections on its antipodal coordinates give the even lengths
$N=q+1,q+3,\ldots,2q$.
\end{theorem}

\smallskip
A computational-basis measurement of any one party gives the
stabilizer $\AME(2q-1,q)$ family by the standard projection property
\cite{Helwig2012}; the conditional states are written explicitly below.
The paired measurements use the specific form  of $A_q$ to
preserve the AME condition at successive even lengths. Their proof
and the short-length construction are also given in
Sec.~\ref{supp:sec-all-lengths} of the SM.

We also emphasize the fact,  that the AME corresponding to  parameter values $q=3,5,7$ on $N=2q$ were previously known 
\cite{GrasslGulliver2008,GrasslGulliver2009,GrasslRotteler2015}; in the
Huber--Wyderka table's snapshot, $q=9$ is the first unsettled member of the family
\cite{AMETable}.  For the four $N=2q$ parents in Table~\ref{tab:finite-atlas}, the
$\AME(34,17)$ representative is $A_{17}$; the other three matrices
are specified separately in the SM.

\smallskip
\noindent\textit{Why the construction works.---}
The finite-field Fourier transform diagonalizes the circulant
block, reducing its Hermitian condition to scalar identities. Order the nonzero labels by powers of
$\gamma$, so the lower-right $(q-1)\times(q-1)$ block is circulant.
Write its first row as
$(\kappa_0,\ldots,\kappa_{q-2})$ and set
$F_k:=\sum_{\ell=0}^{q-2}\kappa_\ell\gamma^{k\ell}\in\F_{q^2}$, with indices modulo $q-1$. This is the finite-field Fourier transform \cite{Pollard1971} over $\F_{q^2}$,
using the $(q-1)$st roots of unity in $\F_q$.
Summing the finite geometric series gives the eigenvalues
(the calculation is given in SM Sec.~\ref{supp:sec-fourier})
\begin{equation}
 F_{2j}=\gamma^j,\qquad F_{2j+1}=g\gamma^j,\qquad
 0\le j\le\frac{q-3}{2}.
 \label{eq:spectrum}
\end{equation}
For nonzero modes, $F_k\overline{F_{-k}}=-1$. The constant border couples
only to the zero mode, where $F_0=1$. The zero mode and border span a two-dimensional subspace.
The choices $\operatorname{Norm}(b)=2$ and $-b/{\bar b}$ for the corner
make the Gram operator $A_q\overline{A_q}^{\,\mathsf T}$ equal to
$-I_2$ on this subspace, completing the identity
$A_q\overline{A_q}^{\,\mathsf T}=-I_q$.

To prove superregularity, choose $\theta^2=\gamma$ in
$\F_{q^2}$, so $\bar\theta=-\theta$. Partial fractions express
each nonborder row through evaluations at $\theta x$ and $-\theta x$.
A singular $s\times s$ minor would give a nonzero rational function
$f(u)/d(u)$ satisfying the selected row equations, with $\deg f\le s-1$.
The columns therefore determine $d$ and a border column contributes a constant kernel.

Write $f(u)=P(u^2)+uQ(u^2)$ and form
\begin{equation}
 H(v)=P(v)\overline{P(v)}-vQ(v)\overline{Q(v)}\in\F_q[v].
\end{equation}
Here conjugation acts on the polynomial coefficients.
This polynomial has degree at most $s-1$.
If both summands are nonzero, their degrees have opposite parity,
so they cannot cancel. Since $f\ne0$, at least one summand
is nonzero, and hence $H\ne0$.
Selecting both rows $x,-x$ forces $P(\gamma x^2)=Q(\gamma x^2)=0$,
giving a double root of $H$. A single row imposes opposite norms
on $f(\theta x)$ and $f(-\theta x)$, giving a root at $\gamma x^2$.
The border row supplies $f(0)=0$, hence a root at zero and 
distinct sign classes give distinct nonzero roots.
Counting multiplicities, the $s$ rows therefore supply at least $s$
roots, contradicting the degree bound. The SM proves this argument
in a form that also covers the mixed rows used in antipodal descent \cite{SM}.

\smallskip
\noindent\textit{The quantum state.---}
The code $\mathcal C_q$ fixes the AME state through its Weyl stabilizers.
Let $p$ be the characteristic of $\F_q$ and use the same
$\theta\in\F_{q^2}^\times$ as above.
Write each coordinate of $\bm c\in\mathcal C_q$ uniquely as
$c_j=x_j+\theta z_j$, with
$x_j,z_j\in\F_q$.  On the computational basis
$\{\ket t:t\in\F_q\}$, for $x,z\in\F_q$, define
\begin{align}
 X(x)\ket t&=\ket{t+x},\nonumber\\
 Z(z)\ket t&=
 \exp\!\left[\frac{2\pi i}{p}
 \Tr_{\F_q/\F_p}(zt)\right]\ket t,\nonumber\\
 D(x,z)&=
 \exp\!\left[\frac{2\pi i}{p}
 \Tr_{\F_q/\F_p}\!\left(\frac{xz}{2}\right)\right]X(x)Z(z).
 \label{eq:weyl}
\end{align}
Here $\Tr_{\F_q/\F_p}(a)=\sum_{\ell=0}^{e-1}a^{p^\ell}$ is
the field trace, whose value is read modulo $p$ in the exponential;
$1/2$ denotes the inverse of $2$ in $\F_q$.
Set $D(\bm c)=\bigotimes_{j=1}^{2q}D(x_j,z_j)$.  For
$\bm c,\bm c'\in\mathcal C_q$, the coefficient of $\theta$ in
$\sum_jc_j\bar c'_j=0$ is
$\sum_j(z_jx'_j-x_jz'_j)=0$, the symplectic product.  Thus Hermitian
self-duality makes the operators commute, while the phase in
Eq.~\eqref{eq:weyl} gives
$D(\bm c)D(\bm c')=D(\bm c+\bm c')$ on $\mathcal C_q$.
Averaging this abelian group of $q^{2q}$ distinct operators gives a
projector. Because  only the identity has nonzero trace, the  projector has
trace one and rank one:
\begin{equation}
 \ket{\Psi_{2q}}\!\bra{\Psi_{2q}}
 =\frac{1}{q^{2q}}\sum_{\bm c\in\mathcal C_q}D(\bm c).
 \label{eq:explicit-projector}
\end{equation}
The distance $q+1$ means that every nonidentity stabilizer acts
on at least $q+1$ sites. If we retain at most $q$ sites, each such
operator acts nontrivially on a discarded site. The local Weyl operator
there has trace zero, so that term vanishes under the partial trace and only the identity remains, giving $\rho_S=I_S/q^{|S|}$.
This proves that the corresponding state defined by the projector in
Eq.~\eqref{eq:explicit-projector} is indeed an $\AME(2q,q)$ state
for every odd prime power $q$ \cite{HuberGrassl2020}.

The first party has reduced state $I_q/q$, so separating it gives
\cite{Helwig2012}
\begin{equation}
 \ket{\Psi_{2q}}
 =\frac{1}{\sqrt q}\sum_{r\in\F_q}
 \ket{r}\otimes\ket{\Psi_{2q-1}^{(r)}},
 \label{eq:projection}
\end{equation}
where the states $\ket{\Psi_{2q-1}^{(r)}}$ are normalized and mutually
orthogonal. Measuring the first party in the computational basis gives
each outcome $r$ with probability $1/q$ and leaves the remaining parties
in $\ket{\Psi_{2q-1}^{(r)}}$.

For any set $T$ of at most $q-1$ survivors, the marginal on
$\{1\}\cup T$ was maximally mixed. Conditioning on outcome $r$ and
dividing by its probability $1/q$ therefore gives
$\rho_T^{(r)}=I_T/q^{|T|}$. Thus every outcome yields a stabilizer
$\AME(2q-1,q)$ state.

\smallskip
\noindent\textit{All lengths and explicit subfamilies.---}
To descend by an even number of parties, choose $r$ disjoint antipodal
pairs $\{x,-x\}$ among the systematic coordinates of $[I_q\mid A_q]$.
Measure the $x$ site in the computational basis and the $-x$
site in the Fourier-conjugate basis, retaining the outcomes
corresponding to $\ket0$ and
$\ket+=q^{-1/2}\sum_{t\in\F_q}\ket t$, respectively.
In the stabilizer labels this restricts
$u_x$ to $\theta\F_q$ and $u_{-x}$ to $\F_q$, after which the measured
coordinates are deleted. The remaining additive code has size
$q^{2q-2r}$ and minimum weight at least $q-r+1$.
The bound follows from the same norm-polynomial argument as
superregularity: a constrained pair of rows compresses to one row
with coefficients of opposite norms.
Consequently, for every $0\le r\le(q-1)/2$ we obtain
$\AME(2q-2r,q)$, together with its odd-party projection.
Standard Reed--Solomon/CSS states cover $2\le N\le q$ and complete
Theorem~\ref{thm:main}; the construction and proof are given in
Sec.~\ref{supp:sec-all-lengths} of the SM.

We also construct Hermitian self-dual MDS codes over $\F_{q^2}$
at selected lengths by rescaling rows and columns of submatrices
of $A_q$. This construction acts on classical generator matrices,
whereas the paired projections described above act on quantum states.

Every square submatrix $P$ of $A_q$ inherits superregularity.
For the selected $m\times m$ submatrices, we find nonsingular
diagonal matrices $W,H$ over $\F_q$ satisfying
$PW\bar P^{\mathsf T}=H$.
Surjectivity of the norm map
$\Norm:\F_{q^2}^{\times}\to\F_q^{\times}$
then allows us to choose nonsingular diagonal matrices
$L,D_{\rm col}$ over $\F_{q^2}$ such that
$B=LPD_{\rm col}$ satisfies $B\bar B^{\mathsf T}=-I_m$.
These row and column scalings preserve superregularity.
Thus the systematic generator $[I_m\mid B]$ defines a Hermitian
self-dual MDS code $[2m,m,m+1]_{q^2}$, which gives a stabilizer
$\AME(2m,q)$ state through the same projector construction.
Sign-symmetric submatrices of $A_q$ yield
\begin{equation}
 \AME(4k+2,q),\quad \AME(4k+1,q),
 \qquad 1\le k\le\frac{q-1}{2},
 \label{eq:split-descendants}
 \nonumber
\end{equation}
where $k=(q-1)/2$ recovers the $2q$-party parent. A second construction, indexed
by a subgroup of $\{z\in\F_{q^2}:\Norm(z)=1\}$, gives
$\AME(q+3,q)$ for every odd prime power and a further descendant family
when $q\equiv1\pmod4$.  All even parents admit computational-basis
one-party projections.  Complete formulas and proofs are given in the
SM \cite{SM}. The weighted constructions preserve linearity over $\F_{q^2}$.
Antipodal descent uses the more general additive stabilizer structure,
which, together with the minimum-weight bound above, establishes the AME property.

\smallskip
\noindent\textit{Multiunitary matrices.---}
It is an interesting observation, that the rescaled amplitudes of an $\AME(2q,q)$ state form a matrix $U_q$ of
order $q^q$, unitary under every balanced reshaping of its indices
\cite{Goyeneche2015}. In suitable local bases the parent stabilizer states
have equal-modulus amplitudes with $p$th-root phases. Thus they give an
infinite family of Butson complex Hadamard matrices, $q$-unitary after normalization
\cite{Bruzda2024Multi}:
\begin{equation}
 \widetilde H_q=q^{q/2}\widetilde U_q\in BH(q^q,p),
 \qquad q=p^e\text{ odd}.
 \label{eq:butson-family}
\end{equation}
Here $|[\widetilde H_q]_{ab}|=1$ and
$\widetilde H_q\widetilde H_q^\dagger=q^qI$; the tilde identifies the
representative obtained after a local change of basis. The explicit
$27\times27$ example already has this form in its original
computational basis: $H_3=\sqrt{27}\,U_3\in BH(27,3)$, with $U_3$
$3$-unitary. The formulas and the general argument are
given in the SM \cite{SM}.

In these local bases, the representatives have $q^{2q}$ nonzero
amplitudes with prescribed phases. Entangled building blocks also
arise for $\AME(4,6)$, where at least one ``quantum officer'' must be
entangled \cite{BallSimoens2026}.
For our $\AME(2q,q)$ family, every local product basis requires more
than the lower bound of $q^q$ nonzero amplitudes. And indeed, the necessary
condition for minimal support, $d\ge\lceil N/2\rceil+1$, would require
$q\ge q+1$ \cite{Goyeneche2015,Bernal2017}.
Their multiunitary matrices are therefore not equivalent to permutation
matrices under single-party unitaries.

For any bipartition, the coefficients of an AME state can be
arranged into a matrix defining a linear map from the smaller
subsystem to the larger one. After multiplication by
the square root of the smaller subsystem's dimension, this map
preserves inner products and is therefore an isometry.
A tensor with this property for every bipartition is called
a perfect tensor. Such tensors serve
as building blocks in holographic quantum error-correcting networks
\cite{Pastawski2015,Bistron:2024rbh}.
At fixed local dimension $q$, the present tensors have $2q$ indices,
compared with at most $q+1$ in the standard minimal-support Reed--Solomon
construction over $\F_q$ \cite{Raissi2018,Heydeman2018}.
The formulas in Eqs.~\eqref{eq:Aq-explicit} and
\eqref{eq:explicit-projector} therefore provide a uniform description
of perfect tensors whose number of indices grows with the local
dimension.

We stress out, that the assumption of odd characteristic is essential to the construction.
For $\tilde{q}=2^r\ge8$, the analogous kernel with two conjugate
simple poles and a constant border admits no choice of parameters
satisfying
$A_{\tilde{q}}\overline{A_{\tilde{q}}}^{\,\mathsf T}=-I_{\tilde{q}}$,
as shown in the SM \cite{SM}.
Thus an extension to these dimensions requires a different kernel
or border structure. The existence of $\AME(2\tilde{q},\tilde{q})$
states through other constructions is a separate open question.

\begin{acknowledgments}
The AI models Claude Fable 5.1 and GPT-6 Astra assisted with finite-field exploration, the development of the explicit matrix formula and its proof, numerical and Lean verification, literature searches, and
manuscript editing.  The authors directed the work, reviewed the resulting formulas and arguments, and take full responsibility for the results and their presentation.

We thank Markus Grassl and Felix Huber for constructive comments.
We also thank Wojciech Bruzda and Suhail Ahmad Rather for discussions
on AME states.

M.H. acknowledges support from the Research Support Module within the
Excellence Initiative--Research University program at Jagiellonian
University in Krak\'ow.  We acknowledge financial support from the
European Union through ERC Advanced Grant \emph{TAtypic}, Project
No.~101142236.

\end{acknowledgments}

\section*{Data Availability}
The explicit formulas, finite-field conventions, and matrix entries
needed to reconstruct the AME states are provided in the Supplemental
Material, including the all-lengths descent and its proof. Lean and numerical verifications are available on reasonable request.

\bibliography{ame}

\onecolumngrid
\clearpage
\raggedbottom
\setcounter{equation}{0}
\setcounter{section}{0}
\setcounter{table}{0}
\setcounter{figure}{0}
\setcounter{theorem}{0}
\setcounter{secnumdepth}{2}
\renewcommand{\theequation}{S\arabic{equation}}
\renewcommand{\thesection}{S\arabic{section}}
\renewcommand{\thesubsection}{\thesection.\arabic{subsection}}
\makeatletter
\renewcommand{\p@subsection}{}
\makeatother
\renewcommand{\thetable}{S\arabic{table}}
\renewcommand{\thefigure}{S\arabic{figure}}
\renewcommand{\thetheorem}{S\arabic{theorem}}
\renewcommand{\thelemma}{\thetheorem}
\renewcommand{\thecorollary}{\thetheorem}
\renewcommand{\theHequation}{S\arabic{equation}}
\renewcommand{\theHsection}{S\arabic{section}}
\renewcommand{\theHsubsection}{S\arabic{section}.\arabic{subsection}}
\renewcommand{\theHtable}{S\arabic{table}}
\renewcommand{\theHfigure}{S\arabic{figure}}
\providecommand{\theHtheorem}{\arabic{theorem}}
\renewcommand{\theHtheorem}{S\arabic{theorem}}
\pdfbookmark[0]{Supplemental Material}{supplement}

\begin{center}
{\large\bfseries Supplemental Material\par}
\medskip
{\bfseries Absolutely Maximally Entangled States\\
for Odd Prime Power Local Dimension \par}
\medskip
Mykhailo Hontarenko and Karol \.{Z}yczkowski
\end{center}

\begin{quotation}
\small
We describe how the finite-field matrix determines an Absolutely Maximally Entangled (AME) state and
work out the six-qutrit example in full. The general construction is
then proved through its Hermitian identity and the invertibility of
its square submatrices. Selected measurements give every length
$2\le N\le2q$, and weighted submatrices provide further explicit
realizations. The remaining sections develop the amplitude-matrix
description and collect the finite-field data and verification results.
\end{quotation}

Sections~\ref{supp:sec-codes} and \ref{supp:sec-qutrit-example} explain
the construction and its first example.
Sections~\ref{supp:sec-family} and \ref{supp:sec-matrix-proof} give the
general matrix and its proof. Shorter states are constructed in
Secs.~\ref{supp:sec-all-lengths} and \ref{supp:sec-weighted};
Table~\ref{supp:tab-family-summary} summarizes their parameters.
The amplitude matrices and the characteristic-two case are discussed
in Secs.~\ref{supp:sec-Butson} and \ref{supp:sec-char-two}.
{Sections~\ref{supp:sec-data} and \ref{supp:sec-map-sources}
collect the compact reconstruction data, exact checks, and map sources.}

\section{From a classical code to an AME state}
\label{supp:sec-codes}

An essential step in our construction is the passage from a
classical code to a stabilizer AME state. We describe this passage
for a matrix $B$ of order $m$ at local dimension $q$.
The universal family corresponds to $B=A_q$ and $m=q$;
the weighted constructions use the same argument at smaller orders:
\begin{equation}
B\ \longrightarrow\ [I_m\mid B]\ \longrightarrow\
 [2m,m,m+1]_{q^2}\ \longrightarrow\ [[2m,0,m+1]]_q
 \ \longrightarrow\ \AME(2m,q).    
\end{equation}

Hermitian self-duality supplies commuting stabilizers, and the minimum
distance determines the maximally mixed marginals
\cite{Scott2004,Review2026}.
Let $q=p^e$, with $p$ an odd prime and $e\ge1$. We work over the
finite field $\F_{q^2}$ and use
\begin{equation}
 \bar z=z^q,\qquad \Norm(z)=z\bar z\in\F_q,
 \label{supp:eq-conjugation}
\end{equation}
extending the bar entrywise to matrices and coefficientwise to
polynomials. Thus $\overline B^{\,\mathsf T}$ is the finite-field
Hermitian transpose. The symbols $\F_q^\times$ and $I_m$ denote the
nonzero field elements and the identity of order $m$; $I_S$ is the
identity on a set $S$ of quantum particles.
Choose a nonsquare $\alpha\in\F_q^\times$ and $\theta^2=\alpha$ in
$\F_{q^2}$. Then $\{1,\theta\}$ is an $\F_q$ basis and
\begin{equation}
 \bar\theta=-\theta.
 \label{supp:eq-theta}
\end{equation}

\subsection{The classical code}
Assume that $B\in\F_{q^2}^{m\times m}$ is superregular and satisfies
$B\overline B^{\,\mathsf T}=-I_m$. Define the $m\times2m$ generator matrix $G_B$ by
\begin{equation}
 G_B=[\,I_m\mid B\,],\qquad
 \mathcal C_B=\{\bm uG_B:\bm u\in\F_{q^2}^{m}\}.
 \label{supp:eq-code}
\end{equation}
The identity block gives rank $m$, and
\begin{equation}
 G_B\overline{G_B}^{\,\mathsf T}
 =I_m+B\overline B^{\,\mathsf T}=0.
 \label{supp:eq-self-orthogonal}
\end{equation}
Hence $\mathcal C_B$ is Hermitian self-dual in $\F_{q^2}^{2m}$.
The MDS criterion requires every $m$ columns of $G_B$ to be independent
\cite[Theorem~2.4.3]{HuffmanPless2003}. Select $m-s$ identity columns
and $s$ columns of $B$, with index set $J$. Expansion along the identity
columns leaves $\pm\det B[R,J]$, where $R$ consists of the $s$
unoccupied row indices. Superregularity makes this determinant nonzero;
for $s=0$ the determinant is $\pm1$ \cite{Keri2006}.
A nonzero word $\bm uG_B$ therefore has at most $m-1$ zero coordinates:
$m$ zeros would give an invertible linear system forcing $\bm u=0$.
Its weight is at least $m+1$, and the Singleton bound gives equality.
Thus the classical parameters of the code are
\begin{equation}
 [2m,m,m+1]_{q^2}.
 \label{supp:eq-classical-code}
\end{equation}
The entries specify length, dimension over $\F_{q^2}$, and minimum
Hamming distance. The Hermitian stabilizer construction gives the pure
quantum MDS code $[[2m,0,m+1]]_q$ \cite{Rains1999,Ketkar2006}.
Here $2m$ counts physical $q$-level particles and the code space has
dimension $q^0=1$. Its pure-state distance is the minimum nonzero
weight of $\mathcal C_B$. We construct the corresponding state next
\cite{HuberGrassl2020}.

\subsection{Stabilizers and maximally mixed marginals}
Write each field coordinate uniquely as $x+\theta z$, with
$x,z\in\F_q$. On the basis $\{\ket t:t\in\F_q\}$ define
\begin{align}
 X(x)\ket t&=\ket{t+x},
 \label{supp:eq-X}\\
 Z(z)\ket t&=
 \exp\!\left[\frac{2\pi i}{p}\Tr_{\F_q/\F_p}(zt)\right]\ket t,
 \label{supp:eq-Z}\\
 D(x,z)&=
 \exp\!\left[\frac{2\pi i}{p}
 \Tr_{\F_q/\F_p}\!\left(\frac{xz}{2}\right)\right]X(x)Z(z).
 \label{supp:eq-Weyl}
\end{align}
Here $\Tr_{\F_q/\F_p}(a)=\sum_{\ell=0}^{e-1}a^{p^\ell}$ is the
field trace, read modulo $p$ in the exponential, and $1/2$ is the
inverse of $2$ in $\F_q$. We reserve $\tr$ for the operator trace.

For $\bm c=(x_j+\theta z_j)_{j=1}^{2m}$ put
$D(\bm c)=\bigotimes_jD(x_j,z_j)$.
The coefficient of $\theta$ in
$\sum_jc_j\bar c'_j=0$ is
$\sum_j(z_jx'_j-x_jz'_j)=0$. Together with the phase convention above,
this gives
\begin{equation}
 D(\bm c)D(\bm c')=D(\bm c+\bm c')
 \qquad(\bm c,\bm c'\in\mathcal C_B).
 \label{supp:eq-Weyl-addition}
\end{equation}
The $q^{2m}$ Weyl operators therefore form an abelian group.
Its average projects onto their common $+1$ eigenspace:
\begin{equation}
 \Pi_B=q^{-2m}\sum_{\bm c\in\mathcal C_B}D(\bm c)
 =q^{-2m}\sum_{\bm u\in\F_{q^2}^{m}}D(\bm u[I_m\mid B]).
 \label{supp:eq-general-projector}
\end{equation}
Multiplication by a group element permutes the summands, and
$D(\bm c)^\dagger=D(-\bm c)$, so the average is an orthogonal
projector. Only the identity has nonzero trace, equal to $q^{2m}$.
Thus $\tr\Pi_B=1$, and $\Pi_B=\ket{\Psi_B}\bra{\Psi_B}$ has rank one.

Every nonidentity term acts on at least $m+1$ sites, by the classical
distance. If $|S|\le m$, it acts nontrivially on a discarded site,
where the local Weyl operator has trace zero. The partial trace retains
only the identity:
\[
 \rho_S=q^{-2m}\tr_{S^c}I=\frac{I_S}{q^{|S|}}.
\]
This proves that $\ket{\Psi_B}$ is $\AME(2m,q)$.
For $B=A_q$ and $m=q$, we write
$G_q=G_B$, $\mathcal C_q=\mathcal C_B$, and
\begin{equation}
 \Pi_q=\ket{\Psi_{2q}}\bra{\Psi_{2q}}
 =q^{-2q}\sum_{\bm c\in\mathcal C_q}D(\bm c).
 \label{supp:eq-projector}
\end{equation}
Sections~\ref{supp:sec-family} and \ref{supp:sec-matrix-proof}
supply $A_q$ and establish its two matrix properties.

\subsection{A one-party measurement}
For any $\AME(2m,q)$ state with $m\ge2$, the first marginal is $I_q/q$.
Consequently
\begin{equation}
 \ket{\Psi_B}=\frac1{\sqrt q}\sum_{r\in\F_q}
 \ket r\otimes\ket{\Psi_{B,-}^{(r)}},
 \label{supp:eq-one-site-projection}
\end{equation}
where the conditional states are normalized and mutually orthogonal.
Each computational-basis outcome has probability $1/q$
\cite{Helwig2012}. For a set $T$ of at most $m-1$ survivors,
$\rho_{\{1\}\cup T}=(I_q\otimes I_T)/q^{|T|+1}$. Taking its
$r$th diagonal block and dividing by the outcome probability gives
\begin{equation}
 \rho_T^{(r)}=q\,\frac{I_T}{q^{|T|+1}}
 =\frac{I_T}{q^{|T|}}.
 \label{supp:eq-projected-marginal}
\end{equation}
Every outcome therefore yields $\AME(2m-1,q)$. Relabeling treats any
measured party, and computational-basis postselection preserves the
stabilizer property. At $B=A_q$, these are the states
$\ket{\Psi_{2q-1}^{(r)}}$ of the main text.

\subsection{Amplitudes and balanced cuts}
\label{supp:sec-multiunitary}
Choose a basis string $\bm t_0$ with
$p_0=\bra{\bm t_0}\Pi_B\ket{\bm t_0}>0$; such a string exists because
$\tr\Pi_B=1$. Fixing the global phase by this choice gives
\begin{equation}
 \ket{\Psi_B}=\frac{\Pi_B\ket{\bm t_0}}{\sqrt{p_0}}
 =\sum_{\bm t\in\F_q^{2m}}\psi_{\bm t}\ket{\bm t},
 \qquad
 \psi_{\bm t}=\frac{\bra{\bm t}\Pi_B\ket{\bm t_0}}{\sqrt{p_0}}.
 \label{supp:eq-amplitudes-from-projector}
\end{equation}
For a set $S$ of $m$ sites, group its indices into
$\bm a\in\F_q^m$ and the complementary indices into $\bm b\in\F_q^m$.
The AME marginal then gives
\begin{equation}
 [U_B^{(S)}]_{\bm a,\bm b}=q^{m/2}\psi_{\bm a,\bm b},
 \qquad U_B^{(S)}(U_B^{(S)})^\dagger=q^m\rho_S=I_{q^m}.
 \label{supp:eq-multiunitary}
\end{equation}
Unitarity for every balanced reshaping is called $m$-unitarity
\cite{Goyeneche2015,Review2026}. Here $\dagger$ is the complex adjoint.
For $B=A_q$, we use $U_q^{(S)}=U_B^{(S)}$ and
$U_q=U_q^{(\{1,\ldots,q\})}$.
The field matrix $A_q$ has order $q$, while the complex amplitude
matrix $U_q$ has order $q^q$.

\section{A worked example: six qutrits}
\label{supp:sec-qutrit-example}

For $q=3$, each physical particle has basis states
$\ket0,\ket1,\ket2$. The field $\F_9=\F_3[\theta]$, with
$\theta^2=-1$, supplies the pairs of labels needed for its Weyl operators.
We shall follow the construction from a matrix of order three to the
amplitudes of a six-particle state.

\subsection{The matrix and the code}
With $g=1+\theta$, the matrix in Eq.~\eqref{eq:A3} is
\[
 A_3=\begin{pmatrix}
 -\theta&1+\theta&1+\theta\\
 1+\theta&1-\theta&\theta\\
 1+\theta&\theta&1-\theta
 \end{pmatrix}.
\]
Here $g$ has order eight and satisfies $g^2+g+2=0$. Equivalently,
\begin{equation}
 A_3=
 \begin{pmatrix}
g^{2}&g&g\\
g&g^{3}&g^{6}\\
g&g^{6}&g^{3}
 \end{pmatrix}.
 \label{supp:eq-A3}
\end{equation}

The identities $g^2=-\theta$, $g^3=1-\theta$, and $g^6=\theta$
relate the two forms. In the bordered-circulant notation introduced in
Sec.~\ref{supp:sec-family}, the parameters are
$a=g^2$, $b=g$, and $\kappa=(g^3,g^6)$.
The generator $G_3=[I_3\mid A_3]$ gives a Hermitian self-dual
$[6,3,4]_9$ code. Its quantum code has parameters $[[6,0,4]]_3$:
six physical qutrits and a one-dimensional code space.
Every nonidentity stabilizer has weight at least four, so the argument
of Sec.~\ref{supp:sec-codes} gives $\rho_S=I_{27}/27$ for any three
particles.

\subsection{The stabilizer labels}
To obtain the amplitudes, separate the field coefficients as
\begin{equation}
 A_3=B+\theta C,\qquad
 B=\begin{pmatrix}0&1&1\\1&1&0\\1&0&1\end{pmatrix},\qquad
 C=\begin{pmatrix}2&1&1\\1&2&1\\1&1&2\end{pmatrix}.
 \label{supp:eq-A3-components}
\end{equation}
All matrices and vectors in the following algebra are over $\F_3$;
vectors are row vectors.  The matrices $B$ and $C$ are symmetric and
commute, and
\begin{equation}
 R:=C^{-1}=\begin{pmatrix}0&2&2\\2&0&2\\2&2&0\end{pmatrix},\qquad
 L:=BR=\begin{pmatrix}1&2&2\\2&2&1\\2&1&2\end{pmatrix}.
 \label{supp:eq-A3-inverse}
\end{equation}
For $\bm u=\bm r+\theta\bm s\in\F_9^3$, the codeword
$\bm c=\bm u[I_3\mid A_3]=\bm x+\theta\bm z$ has
\begin{equation}
 \bm x=(\bm r,\bm rB-\bm sC),\qquad
 \bm z=(\bm s,\bm rC+\bm sB).
 \label{supp:eq-A3-symplectic-coordinates}
\end{equation}
Since $C$ is invertible, $\bm x$ ranges over all of $\F_3^6$.
Eliminating $\bm s$ gives
\begin{equation}
 \bm z=\bm x\Gamma_3,\qquad
 \Gamma_3=\begin{pmatrix}L&-R\\-R&-L\end{pmatrix}
 =\Gamma_3^{\mathsf T}.
 \label{supp:eq-A3-quadratic-matrix}
\end{equation}
Here the upper-right block follows from $C+BRB=-R$, which is checked
directly from Eqs.~\eqref{supp:eq-A3-components} and
\eqref{supp:eq-A3-inverse}.

\subsection{The state and its amplitude matrix}
Let $\zeta=\exp(2\pi i/3)$ be a complex cube root of unity.
The normalized state fixed by these stabilizers is
\begin{equation}
 \ket{\Psi_6}
 =\frac1{27}\sum_{\bm t\in\F_3^6}
 \zeta^{\frac12\bm t\Gamma_3\bm t^{\mathsf T}}\ket{\bm t}.
 \label{supp:eq-A3-quadratic-state}
\end{equation}
The exponent is evaluated modulo three, with $1/2=2$ in $\F_3$.
To verify the phase convention, write
$Q(\bm t)=\bm t\Gamma_3\bm t^{\mathsf T}/2$ and
$\bm z=\bm x\Gamma_3$. Symmetry gives the relation
\[
 Q(\bm t+\bm x)-Q(\bm t)
 =\bm z\bm t^{\mathsf T}+\tfrac12\bm x\bm z^{\mathsf T}.
\]
This is precisely the phase acquired by $\ket{\bm t}$ under
$\bigotimes_jD(x_j,z_j)$ in Eq.~\eqref{supp:eq-Weyl}; hence every such
operator fixes Eq.~\eqref{supp:eq-A3-quadratic-state}.  The rank-one
property of Eq.~\eqref{supp:eq-projector} identifies it with the state
constructed from $A_3$.

Finally split $\bm t=(\bm a,\bm b)$ into two triples, ordered
lexicographically in $\F_3^3$.  Equations
\eqref{supp:eq-multiunitary} and \eqref{supp:eq-A3-quadratic-state} give
a formula for every entry of the $27\times27$ matrix:
\begin{equation}
 [U_3]_{\bm a,\bm b}
 =\frac1{\sqrt{27}}
 \zeta^{\frac12(\bm aL\bm a^{\mathsf T}
                  -\bm bL\bm b^{\mathsf T})
              -\bm aR\bm b^{\mathsf T}},
 \qquad \bm a,\bm b\in\F_3^3.
 \label{supp:eq-U3-explicit}
\end{equation}
All $729$ entries have modulus $1/\sqrt{27}$ and their unnormalized
phases are cubic roots of unity.  Hence
\begin{equation}
 H_3=\sqrt{27}\,U_3\in BH(27,3).
 \label{supp:eq-H3-Butson}
\end{equation}
Here $BH(n,k)$ denotes the Butson class: its entries are $k$th roots of
unity and $HH^\dagger=nI_n$ \cite{Tadej2006,HadamardCatalogue}.
Since the state is $\AME(6,3)$, $U_3$ is \emph{3-unitary}: all
$\binom63=20$ balanced reshufflings are unitary.
Complementary choices transpose the same matrix, so this is
equivalent to ten unordered $3|3$ cuts. Thus $H_3$ is a
multiunitary complex Hadamard matrix in the sense of
Ref.~\cite{Bruzda2024Multi}, with the normalization understood. Figure~\ref{supp:fig-U3-phases} displays the three phases of the explicit
matrix in Eq.~\eqref{supp:eq-U3-explicit}.

\begin{figure}[tbp]
\centering
\includegraphics[width=0.6\linewidth]{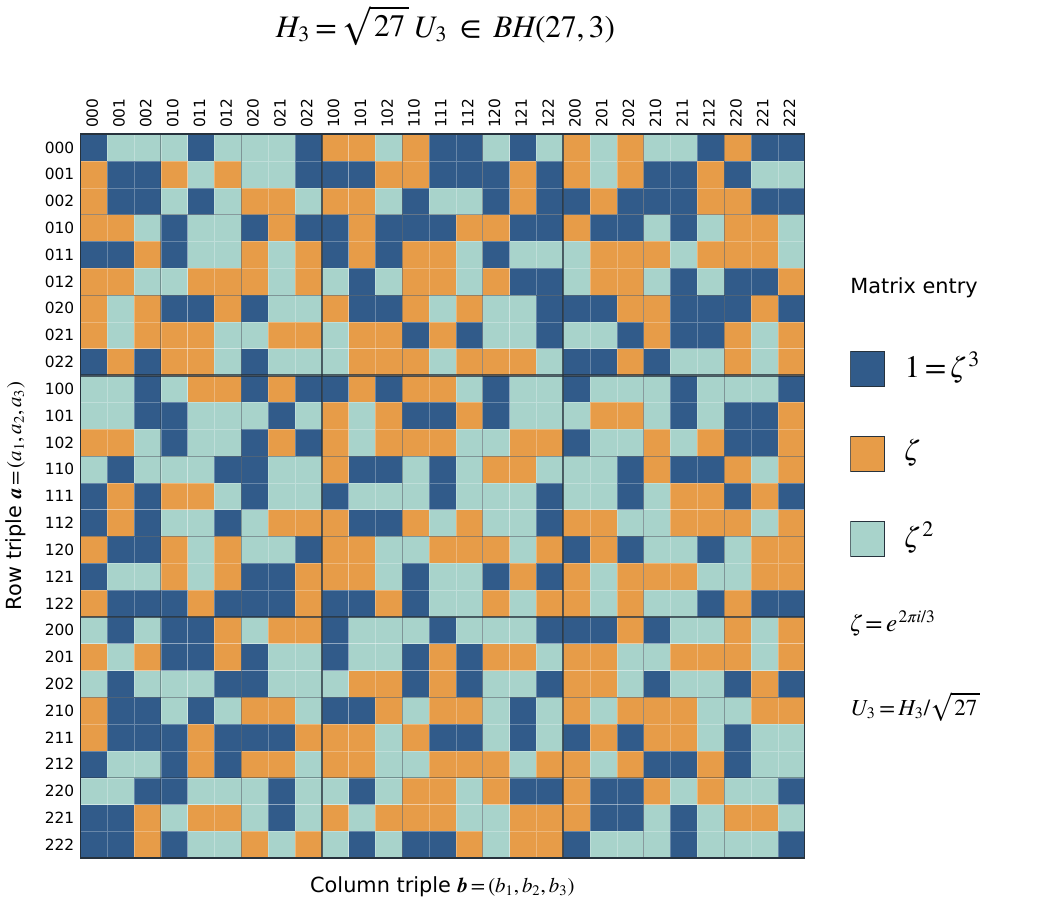}
\caption{The Butson--Hadamard matrix of order $27$,
$H_3=\sqrt{27}\,U_3\in BH(27,3)$, shown in three colors
for the stabilizer $\AME(6,3)$ state. Colors encode the entries
$1$, $\zeta$, and $\zeta^2$, where $\zeta=e^{2\pi i/3}$.
Rows and columns are indexed by triples in $\F_3^3$, in lexicographic
order from $000$ to $222$. All entries of $U_3$ have modulus
$1/\sqrt{27}$; the corresponding state amplitudes have modulus $1/27$.
The normalized matrix $U_3$ is $3$-unitary: all $20$ balanced
reshufflings of its six indices are unitary.}
\label{supp:fig-U3-phases}
\end{figure}

\section{The universal bordered-circulant family}
\label{supp:sec-family}

We now give a matrix with the two properties used in
Sec.~\ref{supp:sec-codes}. A circulant block describes the nonzero field
labels, and a constant border supplies the remaining row and column.
The resulting matrix has order $q$ and generates a code of length $2q$.

For any positive integer $\mu$, elements $a,b\in\F_{q^2}$, and a vector
$\kappa=(\kappa_0,\ldots,\kappa_{\mu-1})\in\F_{q^2}^{\mu}$, define
the bordered-circulant matrix below, where $\bm 1_\mu$ is the
length-$\mu$ all-ones vector:
\begin{equation}
 \operatorname{Circ}_{\mu}(\kappa)_{ij}
 =\kappa_{(j-i)\bmod\mu}
 \quad(0\le i,j<\mu),\qquad
 B_\mu(a,b;\kappa)=
 \begin{pmatrix}
  a&b\bm 1_\mu^{\mathsf T}\\
  b\bm 1_\mu&\operatorname{Circ}_{\mu}(\kappa)
 \end{pmatrix}.
 \label{supp:eq-border}
\end{equation}
Here $a$ is the corner entry and $b$ is the common entry in the
rest of the first row and first column. The circulant block has size
$\mu\times\mu$, while $B_\mu$ has size $(\mu+1)\times(\mu+1)$.
For our universal family we take
$\mu=q-1$, so $A_q$ is a $q\times q$ matrix. Choose
\begin{equation}
 \alpha\in\F_q^\times\ \text{nonsquare},\quad
 \tau\in\F_{q^2}^{\times}, \Norm(\tau)=\alpha,\quad
 \omega\in\F_q^\times, \operatorname{ord}(\omega)=q-1,
 \label{supp:eq-parameters}
\end{equation}
and $b\in\F_{q^2}^{\times}$ with $\Norm(b)=2$.  The parameter
$\tau$ fixes the kernel, while $\omega$ orders the nonzero field labels.
Here $\alpha=\Norm(\tau)$; we retain $\alpha$ as a shorthand in the
formulas and proof. Write $\kappa=\kappa(\alpha,\tau,\omega)$ for the
length-$(q-1)$ kernel vector and set
\begin{align}
 a&=-b/{\bar b},
 \label{supp:eq-a} \\
 \kappa_\ell&=\frac{2(1+\tau\omega^{-\ell})}
 {\alpha\omega^{-2\ell}-1},\qquad 0\le\ell\le q-2,
 \label{supp:eq-kernel}\\
 A_q(\alpha,\tau,\omega,b)&=
 B_{q-1}\!\left(-\frac{b}{\bar b},b;\kappa(\alpha,\tau,\omega)\right).
 \label{supp:eq-Aq}
\end{align}
For this family, $a=-b/\bar b$ is determined by $b$ and serves as
a shorthand in the proof. The general template in
Eq.~\eqref{supp:eq-border} retains arbitrary $a$ and $b$ for the
other constructions below.
The notation $A_q$ will refer to a fixed admissible choice of these
parameters. The main text uses the
choice $(\alpha,\tau,\omega)=(\gamma,g,\gamma)$, where $g$ is primitive
in $\F_{q^2}$ and $\gamma=g^{q+1}$; see
Corollary~\ref{supp:cor-primitive}.
The norm map $\F_{q^2}^{\times}\to\F_q^\times$ is surjective, so the required
$\tau$ and $b$ exist.  The denominator in Eq.~\eqref{supp:eq-kernel} is
nonzero because $\alpha\omega^{-2\ell}$ is a nonsquare, whereas $1=1^2$ is a
square.  Its numerator is also nonzero: $\tau\notin\F_q$, since the norm of
an element of $\F_q$ is its square.

We index the rows of $A_q$ by
$\{0\}\cup\F_q^\times$ and its columns by
$\{\infty\}\cup\F_q^\times$, in the respective orders
$(0,1,\omega,\ldots,\omega^{q-2})$ and
$(\infty,1,\omega,\ldots,\omega^{q-2})$.
The symbol $\infty$ labels the border column.
We write $[A_q]_{x,y}$ for a single matrix entry. For
$x,y\in\F_q^\times$, the entries are
\begin{align}
 [A_q]_{0,\infty}&=a,&
 [A_q]_{0,y}&=b,&
 [A_q]_{x,\infty}&=b,
 \label{supp:eq-borders}\\
 [A_q]_{x,y}&=\frac{2y(y+\tau x)}{\alpha x^2-y^2}.
 \label{supp:eq-core}
\end{align}
Indeed, taking $x=\omega^i$ and $y=\omega^j$ in
Eq.~\eqref{supp:eq-core} gives $\kappa_{(j-i)\bmod(q-1)}$, in agreement
with the circulant convention in Eq.~\eqref{supp:eq-border}. Equations
\eqref{supp:eq-a}--\eqref{supp:eq-core} are the explicit forms of $A_q$
used throughout the paper.

\begin{theorem}
\label{supp:thm-universal}
Every matrix in Eq.~\eqref{supp:eq-Aq} satisfies
\begin{equation}
 A_q\overline{A_q}^{\,\mathsf T}=-I_q,
 \label{supp:eq-gram}
\end{equation}
and is superregular (totally nonsingular): every nonempty square
submatrix has nonzero determinant.
\end{theorem}

The two assertions are proved in Sec.~\ref{supp:sec-matrix-proof}.
Together with the coding argument of Sec.~\ref{supp:sec-codes}, they
give the stabilizer $\AME(2q,q)$ state. We use the term \emph{superregular} in the
sense of Ref.~\cite{Keri2006}.

\begin{corollary}[Primitive-element form]
\label{supp:cor-primitive}
Let $g$ be any primitive element of $\F_{q^2}^{\times}$, and choose
$b\in\F_{q^2}^{\times}$ with $\Norm(b)=2$.  Put
\begin{equation}
 \gamma=g^{q+1},\qquad
 a=-b/{\bar b},\qquad
 \kappa_\ell=\frac{2(1+g\gamma^{-\ell})}
 {\gamma^{1-2\ell}-1},\qquad 0\le\ell\le q-2.
 \label{supp:eq-chirp}
\end{equation}
Then $A_q=B_{q-1}(a,b;\kappa)$, defined in
Eq.~\eqref{supp:eq-border}, satisfies Theorem~\ref{supp:thm-universal}.
\end{corollary}

\noindent\textit{Proof.}---Since $g$ has order $q^2-1$, the order of
$g^{q+1}$ is
$(q^2-1)/(q+1)=q-1$.  Thus $\gamma\in\F_q^\times$ is primitive and,
because $q-1$ is even, is a nonsquare.  Moreover
$\Norm(g)=g^{q+1}=\gamma$.  Taking
$(\alpha,\tau,\omega)=(\gamma,g,\gamma)$ in
Eq.~\eqref{supp:eq-parameters} proves the claim. \hfill$\square$

\section{Proof of the matrix properties}
\label{supp:sec-matrix-proof}

The proof has two parts. The Fourier transform evaluates the Hermitian
Gram matrix. A polynomial argument then establishes superregularity.
Both parts use the same rational expression for the matrix entries.

\subsection{Hermitian orthogonality}
\label{supp:sec-fourier}

The circulant core has $q-1$ rows. We diagonalize it with the
finite-field Fourier transform on $\F_{q^2}^{q-1}$, whose entries are
$\omega^{k\ell}$. Here $\omega\in\F_q$ has order $q-1$.
All entries of this Fourier transform, its eigenvectors, and its
eigenvalues belong to $\F_{q^2}$. For the kernel, the transform is
\begin{equation}
 F_k=\sum_{\ell=0}^{q-2}\kappa_\ell\omega^{k\ell},
 \qquad k\in\mathbb Z/(q-1)\mathbb Z.
 \label{supp:eq-dft}
\end{equation}

The integer $q-1$ represents $-1\ne0$ in $\F_q$, so the inverse
transform exists:
$\kappa_\ell=(q-1)^{-1}\sum_{k=0}^{q-2}F_k\omega^{-k\ell}$.
We use ``Fourier coefficients'' below in this finite-field sense.

\begin{lemma}[Finite-field Fourier coefficients]
\label{supp:lem-fourier}
The kernel in Eq.~\eqref{supp:eq-kernel} has Fourier coefficients
\begin{equation}
 F_{2j}=\alpha^j,\qquad F_{2j+1}=\tau\alpha^j,\qquad
 0\le j\le\frac{q-3}{2}.
 \label{supp:eq-spectrum}
\end{equation}
Consequently Eq.~\eqref{supp:eq-gram} holds.
\end{lemma}

\noindent\textit{Proof.}---The inverse Fourier transform of the right-hand side of
Eq.~\eqref{supp:eq-spectrum} is
\begin{equation}
 (q-1)^{-1}(1+\tau\omega^{-\ell})
 \sum_{j=0}^{(q-3)/2}(\alpha\omega^{-2\ell})^j.
 \label{supp:eq-inverse-dft}
\end{equation}
The ratio $\delta_\ell=\alpha\omega^{-2\ell}$ is a nonsquare, so
$\delta_\ell^{(q-1)/2}=-1$.  Hence the geometric sum is
$-2/(\delta_\ell-1)$.  Since
$q-1=-1$ as an element of $\F_q$, Eq.~\eqref{supp:eq-inverse-dft}
is exactly Eq.~\eqref{supp:eq-kernel}.

Let $C=\operatorname{Circ}_{q-1}(\kappa)$.  On the Fourier vector
$v^{(k)}=(1,\omega^k,\ldots,\omega^{(q-2)k})^{\mathsf T}$,
$C$ has eigenvalue $F_k$, whereas
$\overline C^{\,\mathsf T}$ has eigenvalue
$\overline{F_{-k}}$.  Thus the eigenvalue of
$C\overline C^{\,\mathsf T}$ on this mode is
$F_k\overline{F_{-k}}$.

For a nonzero even mode $k=2j$, with $1\le j\le(q-3)/2$,
the negative mode has index $j'=(q-1)/2-j$, and
\begin{equation}
 F_{2j}\overline{F_{-2j}}
 =\alpha^j\alpha^{(q-1)/2-j}=\alpha^{(q-1)/2}=-1.
 \label{supp:eq-even-pair}
\end{equation}
For an odd mode $k=2j+1$, with $0\le j\le(q-3)/2$,
the negative mode has index $j'=(q-3)/2-j$, giving
\begin{equation}
 F_{2j+1}\overline{F_{-(2j+1)}}
 =\Norm(\tau)\,\alpha^{(q-3)/2}
 =\alpha^{(q-1)/2}=-1.
 \label{supp:eq-odd-pair}
\end{equation}
These formulas also cover the self-inverse mode $k=(q-1)/2$.  Finally,
$F_0=1$, $\Norm(a)=1$, $\Norm(b)=2$, and
\begin{equation}
 a\bar b+b\bar F_0=0.
 \label{supp:eq-border-cancel}
\end{equation}
In the Fourier basis, the zero mode of the lower block of
$A_q\overline{A_q}^{\,\mathsf T}$ is
$\Norm(F_0)+(q-1)\Norm(b)=1-2=-1$, every nonzero mode is $-1$ by
Eqs.~\eqref{supp:eq-even-pair} and \eqref{supp:eq-odd-pair}, and the
upper-left entry is $\Norm(a)+(q-1)\Norm(b)=1-2=-1$.
Equation~\eqref{supp:eq-border-cancel} removes the off-diagonal blocks.
This proves Eq.~\eqref{supp:eq-gram}. \hfill$\square$

\subsection{\texorpdfstring{{Superregularity and a mixed-row lemma}}{Superregularity and a mixed-row lemma}}
\label{supp:sec-factorization}
\label{supp:sec-polynomial}
\label{supp:sec-parity}

{
The same root-counting argument proves superregularity and the distance
bound needed for antipodal descent. We formulate it for the row
combinations that occur in both constructions.
Use $\theta^2=\alpha$ and $\bar\theta=-\theta$ as in
Eq.~\eqref{supp:eq-theta}, and set
\begin{equation}
 \eta_+=\frac{\theta+\tau}{2\theta},\qquad
 \eta_-=\frac{\tau-\theta}{2\theta}.
 \label{supp:eq-etas}
\end{equation}
Since $\Norm(\theta)=-\alpha$ and $\Norm(\tau)=\alpha$, both
coefficients are nonzero. They satisfy
\begin{equation}
 \eta_+-\eta_-=1,\qquad
 \eta_++\eta_-=\frac{\tau}{\theta},\qquad
 \Norm(\eta_+)=-\Norm(\eta_-).
 \label{supp:eq-eta-properties}
\end{equation}
Indeed, $\Norm(\tau-\theta)=\theta(\tau-\bar\tau)$ and
$\Norm(\tau+\theta)=-\theta(\tau-\bar\tau)$.
The partial-fraction identity
\begin{equation}
 [A_q]_{x,y}=2y\left(
 \frac{\eta_+}{\theta x-y}-\frac{\eta_-}{-\theta x-y}\right)
 \label{supp:eq-partial-fraction}
\end{equation}
expresses each column through a rational kernel. Define
\begin{align}
 \mathcal K(v,y)&=(v-y)^{-1}\quad(y\in\F_q^\times),&
 \mathcal K(v,\infty)&=-1,
 \label{supp:eq-cauchy-kernel}\\
 \Lambda_0(\varphi)&=\varphi(0),&
 \Lambda_x(\varphi)&=\eta_+\varphi(\theta x)-\eta_-\varphi(-\theta x).
 \label{supp:eq-functionals}
\end{align}
The matrix $L_{i,y}=\Lambda_i(\mathcal K(\mathord\cdot,y))$, with
rows $\{0\}\cup\F_q^\times$ and columns
$\{\infty\}\cup\F_q^\times$, gives
\begin{equation}
 A_q=\Delta_{\rm row}L\Delta_{\rm col},\qquad
 \Delta_{\rm row}=\operatorname{diag}(-b/2,1,\ldots,1),\qquad
 (\Delta_{\rm col})_{\infty,\infty}=-b,\quad
 (\Delta_{\rm col})_{y,y}=2y.
 \label{supp:eq-factorization}
\end{equation}
For the border, $L_{0,y}=-1/y$ and
$L_{x,\infty}=L_{0,\infty}=-1$; the corner becomes
$-b^2/2=-b/\bar b$. All diagonal factors are invertible, so
$A_q$ and $L$ have the same pattern of vanishing minors.

\begin{lemma}[Mixed rows]
\label{supp:lem-mixed-rows}
From each of several distinct sign classes $\{x,-x\}$ in
$\F_q^\times$, choose either the pair $\Lambda_x,\Lambda_{-x}$
or a single functional
\begin{equation}
 M_x(\varphi)=a_x\varphi(\theta x)+b_x\varphi(-\theta x),
 \qquad
 \Norm(a_x)+\Norm(b_x)=0,\quad (a_x,b_x)\ne(0,0).
 \label{supp:eq-antipodal-mixed}
\end{equation}
One may also choose $\Lambda_0$ once. If $s$ functionals are chosen,
their values on any $s$ distinct kernels
$\mathcal K(\mathord\cdot,y)$, with
$y\in\F_q^\times\cup\{\infty\}$, form a nonsingular matrix.
\end{lemma}

\noindent\textit{Proof.}---Let $J$ be the selected column set and put
\begin{equation}
 d_J(v)=\prod_{y\in J\setminus\{\infty\}}(v-y),\qquad
 \mathcal V_J=\{f/d_J:f\in\F_{q^2}[v],\ \deg f\le s-1\}.
 \label{supp:eq-denominator}
\end{equation}
The selected kernels form a basis of $\mathcal V_J$:
distinct simple poles make the finite kernels independent, and the
constant kernel completes the basis when $\infty\in J$.
Thus a nonzero null vector would give $\varphi=f/d_J$, with $f\ne0$,
annihilated by every chosen functional.
All evaluation points avoid the poles: $0\notin\F_q^\times$ and
$\pm\theta x\notin\F_q$.
Write the even and odd parts of $f$ as
\begin{equation}
 f(v)=P(v^2)+vQ(v^2),\qquad
 H(U)=P(U)\overline{P(U)}-UQ(U)\overline{Q(U)}\in\F_q[U].
 \label{supp:eq-mixed-norm-polynomial}
\end{equation}
We have
$\deg P\le\lfloor(s-1)/2\rfloor$ and
$\deg Q\le\lfloor(s-2)/2\rfloor$, so $\deg H\le s-1$.
For any nonzero polynomial $R$,
$\deg(R\bar R)=2\deg R$, with nonzero leading coefficient
$\Norm(\operatorname{lc}R)$. Hence the two nonzero terms of $H$
have opposite degree parity. As $f\ne0$, at least one term is
nonzero, and therefore $H\ne0$.
A complete pair imposes the system
\[
 \begin{pmatrix}\eta_+&-\eta_-\\-\eta_-&\eta_+\end{pmatrix}
 \begin{pmatrix}\varphi(\theta x)\\\varphi(-\theta x)\end{pmatrix}=0.
\]
Its determinant is
$\eta_+^2-\eta_-^2=\tau/\theta\ne0$, so
$f(\theta x)=f(-\theta x)=0$.
Adding and subtracting these equations gives
$P(\alpha x^2)=Q(\alpha x^2)=0$.
Since $\alpha x^2\in\F_q$, both conjugate polynomials vanish there
as well. Thus $(U-\alpha x^2)^2$ divides $H$.

For a single mixed row, the norm condition makes both $a_x,b_x$
nonzero. Taking norms in
$a_x\varphi(\theta x)=-b_x\varphi(-\theta x)$ gives opposite
norms for the two values. The denominators have equal norm because
$\overline{d_J(\theta x)}=d_J(-\theta x)$.
Consequently
\begin{equation}
 \Norm(f(\theta x))+\Norm(f(-\theta x))
 =2H(\alpha x^2)=0.
 \label{supp:eq-mixed-root}
\end{equation}
Each singleton therefore supplies a root of $H$.
Finally, $\Lambda_0(\varphi)=0$ gives $P(0)=f(0)=0$, hence $H(0)=0$.

Distinct sign classes give distinct nonzero values $\alpha x^2$.
If there are $h$ pairs, $\sigma$ singletons, and
$\varepsilon\in\{0,1\}$ border rows, then
$s=2h+\sigma+\varepsilon$. The nonzero polynomial $H$ has at least
$s$ roots counted with multiplicity, although $\deg H\le s-1$.
This contradiction proves the lemma. \hfill$\square$

Every square submatrix of $L$ meets the lemma's hypotheses:
a single original row has coefficients $(\eta_+,-\eta_-)$ of
opposite norms, and two rows from one sign class form the original
pair. Thus $L$ is superregular. Equation~\eqref{supp:eq-factorization}
transfers this property to $A_q$; together with
Lemma~\ref{supp:lem-fourier}, this proves
Theorem~\ref{supp:thm-universal}.
}

\section{Antipodal descent and all lengths up to \texorpdfstring{$2q$}{2q}}
\label{supp:sec-all-lengths}

Measuring one particle of an even-party AME state gives its odd-party
partner, as shown in Sec.~\ref{supp:sec-codes}. To reach the shorter
even lengths, we measure pairs of coordinates with opposite labels.
Each selected pair removes two particles and leaves the local dimension
$q$ unchanged.

Let $K=\F_q$, $F=\F_{q^2}$, with $q$ an odd prime power, and use
the parent $A_q$ with $\gamma=g^{q+1}$, $\theta^2=\gamma$, and
$\Norm(b)=2$.  Thus $\bar\theta=-\theta$, $\Norm(g)=\gamma$, and
$A_q\bar A_q^{\mathsf T}=-I_q$.
The opposite labels refer to the systematic coordinates of
$[I_q\mid A_q]$. The explicit kernel of $A_q$ will supply the distance
bound for the retained state.

\begin{theorem}[Antipodal descent]
\label{supp:thm-antipodal}
Let $q$ be an odd prime power.  For $0\le r\le(q-1)/2$, any $r$
distinct sign classes in $K^\times=\F_q^\times$ give a stabilizer
$\AME(2q-2r,q)$ state by the paired measurements below.
Together with one-party projections and the short-length constructions,
this gives stabilizer $\AME(N,q)$ states for every $2\le N\le2q$.
\end{theorem}

\subsection{The measurements and the retained code}
Choose representatives
$x_1,\ldots,x_r$ of the sign classes.  Let
$\mathcal E=\{x_i,-x_i:1\le i\le r\}$ be the set of left systematic
coordinates to be measured and removed.  Define
{
\begin{equation}
 \mathcal C_{\mathcal E}=\left\{(\bm u,\bm uA_q)|_{\mathcal E^c}:
 \bm u\in F^q,\quad u_{x_i}\in\theta K,\quad u_{-x_i}\in K
 \ (1\le i\le r)\right\}.
 \label{supp:eq-antipodal-code}
\end{equation}
}
Here restriction deletes those $2r$ left coordinates and retains
\emph{all} right coordinates.  Each line constraint removes one
$K$-dimension; deletion is injective because $\bm uA_q=0$ implies
$\bm u=0$.  Hence $\dim_K\mathcal C_{\mathcal E}=n=2q-2r$.
Hermitian isotropy of the parent implies relative symplectic isotropy.
At a deleted site both labels lie on the same $K$-line, whose symplectic
contribution is zero, so $\mathcal C_{\mathcal E}$ is also isotropic.
Using Sec.~\ref{supp:sec-codes}, its $q^n$ Weyl labels therefore define
the rank-one projector
$\Pi_{\mathcal E}=q^{-n}\sum_{\bm c\in\mathcal C_{\mathcal E}}D(\bm c)$.
The $K$-linear additive structure is sufficient for this stabilizer
construction. It also describes the measurement directly: postselect
$\bra0$ at each $x_i$ and
$\bra+=q^{-1/2}\sum_{t\in K}\bra t$ at each $-x_i$.
Indeed, $\bra0D(a,c)\ket0=\delta_{a,0}$ and
$\bra+D(a,c)\ket+=\delta_{c,0}$.
Let $\mathcal M$ be the product of these bras, acting as the
identity on every retained particle.  It retains exactly the labels in
Eq.~\eqref{supp:eq-antipodal-code}.  With the parent projector $\Pi_q$
from Eq.~\eqref{supp:eq-projector}, injectivity of deletion gives
$\mathcal M\Pi_q\mathcal M^\dagger=q^{-2r}\Pi_{\mathcal E}$.
The branch therefore has probability $q^{-2r}$, and multiplying the
projected parent vector by $q^r$ normalizes it.

\subsection{The distance of the retained code}
{
Apply Lemma~\ref{supp:lem-mixed-rows} with
$(\alpha,\tau)=(\gamma,g)$. The coefficients in
Eq.~\eqref{supp:eq-etas} then satisfy
\[
 \Norm(\eta_+)+\Norm(\eta_-)=0,\qquad
 \eta_+^2-\eta_-^2=g/\theta\ne0.
\]
It remains to check that each measured pair produces an admissible
mixed row, and to translate the lemma into a weight bound.
}

For an active deleted pair
write $c=u_x\in\theta K$ and $d=u_{-x}\in K$.
The row $c\Lambda_x+d\Lambda_{-x}$ has coefficients
$a=c\eta_+-d\eta_-$ and $b'=d\eta_+-c\eta_-$.
They are not both zero, and
\begin{equation}
 \Norm(a)+\Norm(b')=
 -(c\bar d+\bar c d)
 (\eta_+\bar\eta_-+\bar\eta_+\eta_-)=0,
 \label{supp:eq-antipodal-compression}
\end{equation}
since $\bar c=-c$ and $\bar d=d$.
Thus a deleted pair compresses to one admissible mixed row.
Suppose a nonzero word in $\mathcal C_{\mathcal E}$ has weight at
most $q-r$.
Let $k$ count nonzero undeleted left entries, $h$ nonzero right entries,
and $t\le r$ active deleted pairs.  Then
\begin{equation}
 k+h\le q-r,\qquad m:=k+t\le q-h.
 \label{supp:eq-antipodal-count}
\end{equation}
Remove the invertible column scalings, compress the $t$ pairs, and
retain the $k$ original rows, absorbing the border factor $-b/2$
into its coefficient.  Single original rows satisfy the same norm
condition; two surviving rows in a sign class form an original pair.
{The resulting nontrivial combination of $m>0$ admissible
rows vanishes on $q-h\ge m$ right columns, contradicting
Lemma~\ref{supp:lem-mixed-rows}.}
Thus every nonzero label has weight at least $q-r+1$.
Tracing $\Pi_{\mathcal E}$ to any $q-r=n/2$ or fewer sites leaves
only its identity term, by the same argument as in
Sec.~\ref{supp:sec-codes}.  This proves the even-party assertion and
gives a pure quantum MDS code with parameters
$[[2q-2r,0,q-r+1]]_q$.

\subsection{Completion of the all-length family}
{
The paired measurements give the even lengths $q+1,q+3,\ldots,2q$.
The one-party projection proved in Sec.~\ref{supp:sec-codes} gives
their odd partners, preserving the stabilizer property.
For $2\le N\le q$, put $k=\lfloor N/2\rfloor$ and choose distinct
$t_1,\ldots,t_N\in K$. The standard Reed--Solomon state is
}
\begin{equation}
 \ket{\Phi_{N,q}}=q^{-k/2}
 \sum_{a_0,\ldots,a_{k-1}\in K}
 \ket{f(t_1),\ldots,f(t_N)},\qquad
 f(t)=\sum_{j=0}^{k-1}a_jt^j.
 \label{supp:eq-all-lengths-RS}
\end{equation}
Vandermonde invertibility makes evaluation on any $k$ sites bijective
and on their complement injective, giving marginal $I/q^k$.
Translations by this linear evaluation code and phases from its dual
make the state CSS stabilizer.  This completes all $2\le N\le2q$.

\noindent\textit{Formal verification.}---The existence theorem is
formalized as \texttt{AMEVerify.AME\_all\_lengths} in Lean~4.33.1;
the recorded axiom audit lists only Lean's standard axioms.
The formal proof constructs the descendant projectors and verifies
that their reduced states are maximally mixed. The short-length branch of the formal proof uses a Cauchy/CSS code;
Eq.~\eqref{supp:eq-all-lengths-RS} gives the evaluation-code realization
used here. The analytic argument above also establishes the explicit
measurement identities.

\section{Weighted submatrix constructions}
\label{supp:sec-weighted}

The next constructions give explicit Hermitian self-dual codes at
selected lengths. We select rows and columns of a parent matrix and
rescale them over $\F_{q^2}$. The rescaling preserves every nonzero
minor and restores the Hermitian identity needed for the state
construction. Table~\ref{supp:tab-family-summary} collects the resulting
families together with antipodal descent.

\subsection{Normalization and the associated state}
\label{supp:sec-split-descent}

We use the main text's notation $\overline P^{\,\mathsf T}$ for the
finite-field Hermitian transpose, with conjugation as in
Eq.~\eqref{supp:eq-conjugation}.
A square submatrix inherits superregularity from its parent.
The following lemma gives the row and column factors that turn a
weighted diagonal Gram matrix into $-I_m$.

\begin{lemma}[Norm normalization]
\label{supp:lem-weight-normalization}
If a superregular matrix $P\in\F_{q^2}^{m\times m}$ satisfies
$PW\overline P^{\,\mathsf T}=H$, where
$W=\operatorname{diag}(\lambda_j)$ and $H=\operatorname{diag}(h_i)$
with every $\lambda_j,h_i\in\F_q^\times$, choose diagonal matrices
$D_{\rm col},L$ with $\Norm([D_{\rm col}]_{jj})=\lambda_j$ and
$\Norm(L_{ii})=-h_i^{-1}$.  Then $B=LPD_{\rm col}$ is superregular and
$B\overline B^{\,\mathsf T}=-I_m$.
\end{lemma}
\noindent\textit{Proof.}---The norm map is surjective, so these nonzero
entries exist.  Multiplication gives
$B\overline B^{\,\mathsf T}=L\operatorname{diag}(h_i)\overline L^{\,\mathsf T}=-I_m$;
each minor is multiplied by nonzero row and column factors.
\hfill$\square$

{
For the normalized matrix $B$, Sec.~\ref{supp:sec-codes} gives the
$[2m,m,m+1]_{q^2}$ code and the $\AME(2m,q)$ state.
Equations~\eqref{supp:eq-general-projector} and
\eqref{supp:eq-amplitudes-from-projector} specify its projector and
normalized amplitudes. Here $D_{\rm col}$ is a diagonal field matrix;
$D(\bm c)$ denotes a tensor-product Weyl operator.
All odd-party states below follow from the one-party projection
in Eqs.~\eqref{supp:eq-one-site-projection} and
\eqref{supp:eq-projected-marginal}, with outcome probability $1/q$.
}

\subsection{Submatrices with opposite field labels}

Choose the core labels in pairs $\{x,-x\}$ and retain the border.
Selecting $k$ such pairs gives a matrix of order $2k+1$, hence a
state of $4k+2$ parties. The weights below express the normalization
through two polynomials whose roots are the selected squares and
their nonsquare multiples.

\begin{theorem}[Split descent]
\label{supp:thm-split-descent}
Let $q$ be an odd prime power.  For $1\le k\le(q-1)/2$, choose
$x_1,\ldots,x_k\in\F_q^\times$ with
distinct squares $u_i=x_i^2$.  The submatrix
\[
 P=A_q[\{0,\pm x_1,\ldots,\pm x_k\},
       \{\infty,\pm x_1,\ldots,\pm x_k\}]
\]
has a diagonal normalization $B$ with $B\overline B^{\,\mathsf T}=-I_{2k+1}$.
Consequently there are Hermitian self-dual
$[4k+2,2k+1,2k+2]_{q^2}$ MDS codes and stabilizer
$\AME(4k+2,q)$ and $\AME(4k+1,q)$ states.
\end{theorem}

\noindent\textit{Proof.}---Use the parameters $\alpha,\tau,b$ defining
$A_q$ in Eqs.~\eqref{supp:eq-parameters}--\eqref{supp:eq-Aq}, and set
\begin{equation}
 Q(T)=\prod_i(T-u_i),\quad R(T)=\prod_i(T-\alpha u_i),\quad
 \lambda_\infty=1,\quad
 \lambda_{x_i}=\lambda_{-x_i}=\lambda_i
 =-\frac{R(u_i)}{4u_iQ'(u_i)}.
 \label{supp:eq-split-weights}
\end{equation}
These weights are nonzero in $\F_q$: the $u_i$ are distinct nonzero
squares, whereas the $\alpha u_i$ are nonsquares.
Introduce the auxiliary rational function $\mathcal F$.
Partial fractions give
\begin{equation}
 \mathcal F(T)=\sum_i\frac{\lambda_i u_i}{T-u_i}
 =\frac14\left(1-\frac{R(T)}{Q(T)}\right),
 \qquad \mathcal F(\alpha u_i)=\frac14.
 \label{supp:eq-split-interpolation}
\end{equation}
For $u=t^2$, Eq.~\eqref{supp:eq-core} gives the two pair identities
\begin{align*}
 \overline{[A_q]_{x,t}}+\overline{[A_q]_{x,-t}}
 &=\frac{4u}{\alpha x^2-u},\\
 [A_q]_{x,t}\overline{[A_q]_{z,t}}+
 [A_q]_{x,-t}\overline{[A_q]_{z,-t}}
 &=\frac{8u(u+\alpha xz)}{(\alpha x^2-u)(\alpha z^2-u)}.
\end{align*}
Write $W=\operatorname{diag}(1,\lambda_1,\lambda_1,\ldots,
\lambda_k,\lambda_k)$.  The row-$0$/row-$x$ and row-$x$/row-$(-x)$
entries of $PW\overline P^{\,\mathsf T}$ are respectively
$b[-1+4\mathcal F(\alpha x^2)]=0$ and $2-8\mathcal F(\alpha x^2)=0$.
For $x^2\ne z^2$, set $a_x=\alpha x^2$, $a_z=\alpha z^2$,
$s=\alpha xz$.  Decomposing the second pair identity into simple
fractions gives the remaining off-diagonal entry as
\[
 2+8\left[\frac{a_z+s}{a_x-a_z}\mathcal F(a_z)
          -\frac{a_x+s}{a_x-a_z}\mathcal F(a_x)\right]=0.
\]
At the diagonal, evaluation at $T=0$ and differentiation of
Eq.~\eqref{supp:eq-split-interpolation} give
\begin{equation}
 h_0=1+4\sum_i\lambda_i=\alpha^k,\qquad
 h_{\pm x_i}=-16\alpha u_i\mathcal F'(\alpha u_i)
 =4\alpha u_i\frac{R'(\alpha u_i)}{Q(\alpha u_i)}.
 \label{supp:eq-split-weighted-gram}
\end{equation}
Every diagonal entry is nonzero because $R$ has simple roots disjoint
from those of $Q$.  The parent theorem makes $P$ superregular, so
Lemma~\ref{supp:lem-weight-normalization} and
Eq.~\eqref{supp:eq-general-projector} prove the assertions.
At $k=(q-1)/2$, $Q(T)=T^k-1$, $R(T)=T^k+1$, and every weight is $1$;
the parent identity is recovered exactly. \hfill$\square$

\subsection{A parent indexed by the norm-one circle}
\label{supp:sec-circle-parent}

A second choice of labels comes from the norm-one subgroup of
$\F_{q^2}^{\times}$. Its square subgroup gives a bordered matrix of
order $(q+3)/2$, and therefore an $\AME(q+3,q)$ state.

Let $\mathcal U=\{z\in\F_{q^2}^\times:\Norm(z)=1\}$,
$\nu=(q+1)/2$, and $\mathcal H=\mathcal U^2$, a cyclic group of order
$\nu$.  Set $c=1$ when $q\equiv1\pmod4$ and choose
$c\in\mathcal U\setminus\mathcal H$ when $q\equiv3\pmod4$.
Then $(-c)^\nu=-1$, so $x+cy\ne0$ for $x,y\in\mathcal H$.
Choose $u$ of norm $-1$ such that
$t=\Tr_{\F_{q^2}/\F_q}(c\bar u)\notin\{0,-1,2,3/2\}$.
Here the trace is $\Tr_{\F_{q^2}/\F_q}(z)=z+\bar z$.
Put
\begin{equation}
 P_0=c+2u,\qquad \sigma=1+u/c,\qquad
 \Norm(b_\circ)=-2(1+t),\qquad
 a_\circ=-b_\circ\bar\sigma/\bar b_\circ.
 \label{supp:eq-circle-parameters}
\end{equation}
These choices always exist.  On the norm-$-1$ circle each trace fibre
has at most two elements (its elements solve $X^2-tX-1=0$), so at least
$(q+1)/2$ trace values occur.  This exceeds the four exclusions for
$q\ge9$.  For $q=3,5,7$ take $t=1$: choose a root $z$ of
$X^2-X-1$ for $q=3,7$, or $z=3\in\F_5$ for $q=5$, and set
$u=c\bar z$.  Thus the construction includes $q=3$.
Norm surjectivity supplies $b_\circ$, and direct multiplication gives
$\Norm(P_0-c)=-4$, $\Norm(P_0)=2t-3\notin\{0,1\}$,
and $\Norm(\sigma)=t\ne0$.

The circle matrix $B_q^\circ$ has order
$\nu+1=(q+3)/2$. Its rows are labeled by $\{0\}\cup\mathcal H$
and its columns by $\{\infty\}\cup\mathcal H$. Define
\begin{equation}
 [B_q^\circ]_{0,\infty}=a_\circ,\qquad
 [B_q^\circ]_{0,y}=[B_q^\circ]_{x,\infty}=b_\circ,\qquad
 [B_q^\circ]_{x,y}=\frac{2(P_0y+x)}{x+cy}.
 \label{supp:eq-circle-parent}
\end{equation}
This has the bordered-circulant form of Eq.~\eqref{supp:eq-border}.
Choose a generator $h$ of $\mathcal H$ and order both core index sets
as $1,h,\ldots,h^{\nu-1}$. Then
\[
 B_q^\circ=B_\nu(a_\circ,b_\circ;\kappa^\circ),\qquad
 \kappa^\circ_\ell=\frac{2(P_0h^\ell+1)}{1+ch^\ell},
 \quad 0\le\ell<\nu.
\]
Indeed, setting $x=h^i$ and $y=h^j$ gives the entry
$\kappa^\circ_{(j-i)\bmod\nu}$. The core order is now
$\nu=(q+1)/2$, rather than $q-1$ as for $A_q$.

\begin{theorem}[Circle parent]
\label{supp:thm-circle-parent}
The matrix $B_q^\circ$ is superregular and
$B_q^\circ\overline{B_q^\circ}^{\,\mathsf T}=-I_{(q+3)/2}$.
Its systematic generator $[I_{(q+3)/2}\mid B_q^\circ]$ has $q+3$
columns and yields a Hermitian self-dual
$[q+3,(q+3)/2,(q+5)/2]_{q^2}$ MDS code and stabilizer
$\AME(q+3,q)$ and $\AME(q+2,q)$ states for every odd prime power $q$.
\end{theorem}
Thus the circle matrix satisfies the same Hermitian Gram condition
as $A_q$ in Eq.~\eqref{supp:eq-gram}, but at order $(q+3)/2$.

\noindent\textit{Proof.}---For $\rho\in\mathcal H$, put
$f(\rho)=2(P_0\rho+1)/(1+c\rho)$.  The geometric identity
$(1+c\rho)\sum_{j=0}^{\nu-1}(-c\rho)^j=2$ gives
\[
 \widehat f(0)=\sigma,\qquad
 \widehat f(k)=\nu(-c)^{\nu-k-1}(P_0-c)
 \quad(1\le k<\nu),\qquad
 \Norm(\widehat f(k))=-1\quad(k\ne0).
\]
Here $\widehat f(k)=\sum_{\rho\in\mathcal H}f(\rho)\rho^k$ and
$\nu=1/2$ inside $\F_q$.  Since conjugation inverts every element of
$\mathcal H$, the product of the core with its Hermitian transpose has
eigenvalues $\Norm(\widehat f(k))$.  Its zero mode has eigenvalue $t$;
the core itself has row sum $\sigma$.
Adding the border changes the zero-mode eigenvalue to
$t+\nu\Norm(b_\circ)=-1$, leaves all nonzero modes unchanged,
and gives the same upper-left entry.  The off-diagonal border block
vanishes because $a_\circ\bar b_\circ+b_\circ\bar\sigma=0$.
Thus $B_q^\circ\overline{B_q^\circ}^{\,\mathsf T}=-I_{(q+3)/2}$.

To prove superregularity, first consider an $s\times s$ minor without
column $\infty$, with rows $I\subseteq\{0\}\cup\mathcal H$ and
columns $J\subseteq\mathcal H$.  Up to nonzero row factors its entries
are $\mathcal K_\circ(v,y)=(P_0y+v)/(v+cy)$; at row $0$ this constant is
$P_0/c\ne0$.  A nonzero null vector $(e_y)$ would give
\[
 \varphi(v)=\sum_{y\in J}e_y\mathcal K_\circ(v,y)
 =E+(P_0-c)\frac{g(v)}{D(v)},\qquad
 E=\sum_y e_y,\quad D(v)=\prod_{y\in J}(v+cy),\quad \deg g\le s-1,
\]
vanishing at all $s$ selected rows.  Its numerator is nonzero:
otherwise its residues force every $e_y=0$.  Therefore $E\ne0$,
since otherwise that numerator has degree at most $s-1$.
If $0\in I$, the equation $\varphi(0)=P_0E/c=0$ is a contradiction.
Otherwise the numerator equals $E\prod_{x\in I}(v-x)$.
Evaluating at $v=0$ gives
\[
 \frac{\prod_{x\in I}x}{\prod_{y\in J}y}
 =(-1)^s c^{s-1}P_0,
\]
which is impossible: the left side has norm $1$, while the right side
has norm $\Norm(P_0)\ne1$.  All minors without column $\infty$ are
therefore nonzero.  Since $(B_q^\circ)^{-1}=-\overline{B_q^\circ}^{\,\mathsf T}$, Jacobi's
complementary-minor identity transfers this result to all minors using
column $\infty$ (the full determinant is already nonzero).
The projector construction in Eq.~\eqref{supp:eq-general-projector},
with $m=(q+3)/2$, gives the stated AME states. \hfill$\square$

\subsection{Odd subsets of the circle}
\label{supp:sec-circle-descent}

When $q\equiv1\pmod4$, an odd number of core labels can be retained
together with the border. The following weights provide the
normalization for every such subset.

\begin{theorem}[Circle descent]
\label{supp:thm-circle-descent}
Let $q$ be an odd prime power with $q\equiv1\pmod4$.
For every $1\le k\le(q+3)/4$ and every
$S\subseteq\mathcal H$ of size $2k-1$, the submatrix
$P=B_q^\circ[\{0\}\cup S,\{\infty\}\cup S]$, with $B_q^\circ$ defined
in Eq.~\eqref{supp:eq-circle-parent}, admits a superregular
normalization $B$ with $B\overline B^{\,\mathsf T}=-I_{2k}$.
It yields a Hermitian self-dual $[4k,2k,2k+1]_{q^2}$ MDS code and
stabilizer $\AME(4k,q)$ and $\AME(4k-1,q)$ states.
\end{theorem}

\noindent\textit{Proof.}---Now $c=1$ and $-1\notin\mathcal H$.
For the selected submatrix $P$, define scalar weights $r_y$ and
two auxiliary functions:
\begin{equation}
 r_y=\prod_{x\in S\setminus\{y\}}\frac{y+x}{y-x},\qquad
 Q(T)=\prod_{x\in S}(T-x),\qquad
 H(T)=\sum_{y\in S}\frac{r_yy}{T+y}
 =\frac12\left(1+\frac{Q(T)}{Q(-T)}\right).
 \label{supp:eq-circle-barycentric}
\end{equation}
Each factor of $r_y$ changes sign under conjugation; their number is
even, so $r_y\in\F_q^\times$.  The rational identity follows by
matching residues and the value at infinity, using
$r_y=-Q(-y)/(2yQ'(y))$.
Evaluating at $T=0,x$ gives $\sum_y r_y=1$ and $H(x)=1/2$.
Also $H'(x)=-1/(4xr_x)$, whence
\begin{equation}
 \sum_y r_y\frac{y^2}{(x+y)(z+y)}=\frac12\quad(x\ne z),\qquad
 \sum_y r_y\frac{y^2}{(x+y)^2}=\frac12-\frac1{4r_x}.
 \label{supp:eq-circle-moments}
\end{equation}
Indeed, these sums are respectively
$(zH(z)-xH(x))/(z-x)$ and $H(x)+xH'(x)$.
Put $K_x(y)=1+(P_0-1)y/(x+y)$, so
$[B_q^\circ]_{x,y}=2K_x(y)$ and
$\overline{K_z(y)}=\bar P_0-(\bar P_0-1)y/(z+y)$.
Substitution of these moments, $\Norm(P_0-1)=-4$, and
$\sigma=(P_0+1)/2$ gives
\[
 \sum_y r_y\overline{K_x(y)}=\bar\sigma,\qquad
 \sum_y r_yK_x(y)\overline{K_z(y)}
 =1+\Norm(\sigma)-\frac{\delta_{xz}}{r_x}.
\]
Together with $\Norm(b_\circ)=-2(1+\Norm(\sigma))$ and
$a_\circ\bar b_\circ=-b_\circ\bar\sigma$,
these identities show that the weighted Gram matrix of $P$ is diagonal:
\begin{equation}
 P\operatorname{diag}\bigl(1,(r_y/2)_{y\in S}\bigr)\overline P^{\,\mathsf T}
 =\operatorname{diag}\bigl(-1,(-2/r_x)_{x\in S}\bigr).
 \label{supp:eq-circle-weighted-gram}
\end{equation}
Both diagonals are nonsingular over $\F_q$.  Superregularity is inherited
from $B_q^\circ$, and Lemma~\ref{supp:lem-weight-normalization} and
Eq.~\eqref{supp:eq-general-projector} finish the proof. \hfill$\square$

\subsection{Parameter ranges}

For $q\equiv1\pmod4$, the split and circle descendants together cover
all $3\le N\le q+5$. The split family continues in the classes
$N\equiv1,2\pmod4$ up to
$2q$.  In particular, its choice $k=(q+3)/4$ gives $\AME(q+5,q)$.
The circle parent is the top circle descendant when $q\equiv1\pmod4$;
when $q\equiv3\pmod4$, its parameters are already covered by the split
family. Proper circle descendants have even parent length at most
$q-1$. The non-GRS and support results stated for the $2q$-party
parents concern that family. Equivalence among the shorter
representatives is a separate classification question.

Table~\ref{supp:tab-family-summary} collects the Hermitian-linear
families and the additive antipodal descendants. The short-length CSS
construction in Eq.~\eqref{supp:eq-all-lengths-RS} completes their
parameter coverage to every $2\le N\le2q$ for every odd prime power
$q=p^e$, where $p\ge3$.

\begin{table}[tbp]
\caption{Analytic constructions and their AME states. The first four
rows use Hermitian self-dual MDS codes over $\F_{q^2}$; the final row
uses an additive symplectic code over $\F_q$.
Throughout, $q$ is an odd prime power and $k$ is
an integer.  For split descent, $1\le k\le(q-1)/2$.
Circle descent requires $q\equiv1\pmod4$ and
$1\le k\le(q+3)/4$. For antipodal descent,
$m$ runs through the integers $(q+1)/2,\ldots,q$.
The quantum column lists the pure parent-state codes
with zero logical qudits; the last column gives normalized
computational-basis one-party projections.  Parameter ranges overlap,
as explained in the text.}
\label{supp:tab-family-summary}
\small
\setlength{\tabcolsep}{4pt}
\renewcommand{\arraystretch}{1.5}
\begin{ruledtabular}
\begin{tabular}{lcccc}
Construction &
\shortstack{Classical code over $\F_{q^2}$\\or label space} &
Quantum code & Parent state & Projected state \\
\colrule
Universal (Thm.~\ref{supp:thm-universal}) &
$[2q,q,q+1]$ &
$[[2q,0,q+1]]_q$ &
$\AME(2q,q)$ & $\AME(2q-1,q)$ \\
Split (Thm.~\ref{supp:thm-split-descent}) &
$[4k+2,2k+1,2k+2]$ &
$[[4k+2,0,2k+2]]_q$ &
$\AME(4k+2,q)$ & $\AME(4k+1,q)$ \\
Circle (Thm.~\ref{supp:thm-circle-parent}) &
$[q+3,\frac{q+3}{2},\frac{q+5}{2}]$ &
$[[q+3,0,\frac{q+5}{2}]]_q$ &
$\AME(q+3,q)$ & $\AME(q+2,q)$ \\
Circle descent (Thm.~\ref{supp:thm-circle-descent}) &
$[4k,2k,2k+1]$ &
$[[4k,0,2k+1]]_q$ &
$\AME(4k,q)$ & $\AME(4k-1,q)$ \\
Antipodal (Thm.~\ref{supp:thm-antipodal}) &
\shortstack{$\F_q$-linear additive\\$|\mathcal C_{\mathcal E}|=q^{2m}$} &
$[[2m,0,m+1]]_q$ &
$\AME(2m,q)$ & $\AME(2m-1,q)$ \\
\end{tabular}
\end{ruledtabular}
\end{table}

\section{Structure of the amplitude matrices}
\label{supp:sec-Butson}

The six-qutrit example exhibits a useful feature of the whole
construction: suitable local bases give amplitudes of equal modulus.
For every balanced cut, their phases form a complex Hadamard matrix.
We now establish this representation for the families above and
describe its relation to support and stabilizer entropy.

\subsection{Graph-state bases and Butson matrices}

\begin{corollary}[Multiunitary Butson representatives]
\label{supp:cor-Butson}
Every stabilizer $\AME(2m,q)$ state constructed here, with $q=p^e$ odd,
has a locally equivalent representative of the form
\begin{equation}
 \ket{\widetilde\Psi}
 =q^{-m}\sum_{\bm t\in\F_p^{2em}}
 \zeta_p^{Q(\bm t)}\ket{\bm t},
 \qquad \zeta_p=e^{2\pi i/p},
 \label{supp:eq-graph-state}
\end{equation}
where $Q$ is quadratic over $\F_p$ and each group of $e$ prime-dimensional
indices labels one original $q$-dimensional party.  For any balanced cut
of these $2m$ parties, the rescaled amplitude matrix $\widetilde U$ is
unitary, and
\begin{equation}
 \widetilde H=q^{m/2}\widetilde U\in BH(q^m,p)\subseteq BH(q^m,q).
 \label{supp:eq-Butson-family}
\end{equation}
In particular, the universal family gives $q$-unitary complex Hadamard
matrices of order $q^q$ for every odd prime power $q$.
\end{corollary}

\noindent\textit{Proof.}---For odd prime $p$, every stabilizer state is
local-Clifford equivalent to a graph state
\cite[Lemma~7]{BahramgiriBeigi2007}.  For $q=p^e$, choose an $\F_p$ basis
of $\F_q$ and its trace-dual basis.  The finite-field Weyl operators then
act as tensor products on $e$ $p$-dimensional factors per party.  Applying
the prime-dimensional result to these $2em$ factors and regrouping them
therefore uses only operations local with respect to the original
$2m$ parties.  The graph-state amplitudes have the form
Eq.~\eqref{supp:eq-graph-state}. Local operations preserve the AME marginals, so every balanced
reshaping of the amplitudes, multiplied by $q^{m/2}$, is unitary.
Its entries have modulus $q^{-m/2}$, and multiplication by $q^{m/2}$
leaves $p$th roots of unity.  Since $p\mid q$, these are also $q$th roots.
\hfill$\square$

For the universal matrices, this basis change is essential when $q\ge5$.
Write $A_q=B+\theta C$ with $B,C$ over $\F_q$, so $C$ records
the coefficients of $\theta$.
In the core of $C$, the
columns labeled by $y$ and $-y$ are opposites; their pair sums have the
same border entry.  Differences of two such pair sums give a nonzero
kernel vector for $C$, so $C$ is singular.  A nonzero row vector
$\bm sC=0$ then gives the codeword
$\theta\bm s[I_q\mid A_q]$ with zero $X$ component, hence a nonidentity
pure-$Z$ stabilizer. Its fixed state is supported on the basis labels satisfying that
stabilizer equation. For $q\ge5$, a local change of basis is therefore
needed to obtain the equal-modulus matrix $\widetilde U$ in
Eq.~\eqref{supp:eq-Butson-family}.

\subsection{Hadamard equivalence}

Under ordinary Hadamard equivalence---row and column permutations
and multiplication of rows and columns by unit-modulus phases
\cite{Tadej2006}---each matrix $\widetilde H$ is equivalent to
$F_p^{\otimes em}$, the Fourier matrix of the additive group
$\F_p^{em}$. Here $(F_p)_{ab}=\zeta_p^{ab}$.
To see this, fix a balanced cut and split the quadratic phase into
a row term, a column term, and a bilinear cross term over $\F_p$.
Removing the first two terms by row and column phases leaves entries
$\zeta_p^{\bm a K\bm b^{\mathsf T}}$, with
$\bm a,\bm b\in\F_p^{em}$. Unitarity forces $K$ to be invertible:
a nonzero vector in its left kernel would give two identical rows.
Relabeling the columns by this invertible linear map gives
$F_p^{\otimes em}$. These operations act on the row and column spaces
of the chosen cut. Preserving the multipartite structure requires
the more restrictive single-party operations discussed in
Ref.~\cite{Bruzda2024Multi}.

\subsection{Support in local product bases}
A unitary monomial matrix has exactly one nonzero entry of modulus one
in every row and column; permutation matrices are the special case with
all these entries equal to one.  If $U_q$ were monomial, the corresponding
AME state would have exactly $q^q$ nonzero amplitudes, the minimum
possible support.  However, a necessary condition for a minimally
supported $\AME(N,d)$ state with $N\ge4$ is
$d\ge\lceil N/2\rceil+1$ \cite{Goyeneche2015,Bernal2017}.
For $N=2q$ and $d=q$, this would require $q\ge q+1$.
Hence every local product basis requires more than $q^q$ nonzero
amplitudes for $\AME(2q,q)$. Equivalently, its amplitude matrix lies
outside the monomial class under single-party changes of basis.
The monomial Weyl operators in Eq.~\eqref{supp:eq-projector} serve
a different role: their average is the state's rank-one projector.

\subsection{Stabilizer R\'enyi entropy}
The order-two stabilizer R\'enyi entropy, denoted $M_2$, measures
nonstabilizerness
\cite{Leone2022StabilizerRenyi,Klich2024WeylMagic}. Using the finite-field
Weyl operators of Eq.~\eqref{supp:eq-Weyl}, one defines
\begin{equation}
 M_2(\Psi)=-\log\!\left[
 \frac{1}{q^N}\sum_{\bm x,\bm z\in\F_q^N}
 \left|\bra{\Psi}\bigotimes_{j=1}^{N}D(x_j,z_j)\ket{\Psi}\right|^4
 \right].
 \label{supp:eq-magic-M2}
\end{equation}
For a pure stabilizer state, exactly $q^N$ Weyl labels have expectation
values of modulus one; all remaining expectation values vanish.
Thus $M_2(\Psi)=0$ for every state constructed here, including
$M_2(\Psi_6)=0$ for the example above. The support size discussed above and the stabilizer entropy
therefore describe distinct properties of these states.
Here $M_2$ is defined relative to the chosen Weyl operators and is
preserved by Clifford transformations; a general local unitary can
change its value.

\section{The two-pole kernel in characteristic two}
\label{supp:sec-char-two}

We now test the same constant-plus-two-poles form in characteristic
two. Let $q=2^r$ with $r\ge3$ and choose
$\theta\in\F_{q^2}\setminus\F_q$ with $\bar\theta=\theta+1$.
Consider
\begin{equation}
 f(t)=c_0+\frac{c_1}{t+\theta}
             +\frac{c_2}{t+\bar\theta},
 \qquad t\in\F_q^\times,
 \label{supp:eq-char-two-kernel}
\end{equation}
write
$C_f=(f(y/x))_{x,y\in\F_q^\times}$ and
$\widehat f(k)=\sum_{t\in\F_q^\times}f(t)t^k$, with
$k\in\mathbb Z/(q-1)\mathbb Z$.  This is the direct characteristic-two
counterpart of the constant-plus-two-poles kernel obtained from
Eq.~\eqref{supp:eq-partial-fraction} after the invertible border scalings.

\begin{lemma}[Characteristic-two case]
\label{supp:lem-char-two}
No bordered circulant whose core is $C_f$ with $f$ as in
Eq.~\eqref{supp:eq-char-two-kernel} can satisfy the matrix condition
$A\overline A^{\,\mathsf T}=-I_q$
for $q=2^r\ge8$.
\end{lemma}

\noindent\textit{Proof.}---For $1\le k\le q-2$, a finite geometric sum
gives
\begin{equation}
 \widehat{(t+\theta)^{-1}}(k)=\theta^k,
 \qquad
 \widehat{(t+\bar\theta)^{-1}}(k)=\bar\theta^{\,k}.
 \label{supp:eq-char-two-poles}
\end{equation}
For example,
$(t+\theta)^{-1}=(\bar\theta/\theta)
\sum_{j=0}^{q-2}\theta^{-j}t^j$ on $\F_q^\times$, from which the first
identity follows by character orthogonality; the second is its conjugate.
The constant term in Eq.~\eqref{supp:eq-char-two-kernel} has zero Fourier
transform on these modes, and hence
\begin{equation}
 \widehat f(k)=c_1\theta^k+c_2\bar\theta^{\,k}.
 \label{supp:eq-char-two-spectrum}
\end{equation}
Let us put $\rho=\theta/\bar\theta$.  It is a nontrivial norm-one element, so its
order is an odd divisor of $q+1$ and is at least $3$.  Since
$\theta^{q-1}=\rho^{-1}$ and
$\bar\theta^{\,q-1}=\rho$, direct multiplication yields
\begin{align}
 \widehat f(k)\overline{\widehat f(-k)}
 ={}c_1\bar c_2\rho^{-1}+\bar c_1c_2\rho
 +\Norm(c_1)\rho^{k+1}+\Norm(c_2)\rho^{-(k+1)}.
 \label{supp:eq-char-two-pairing}
\end{align}
If this matrix condition held, the constant border would affect only
the zero mode, while every nonzero mode in
Eq.~\eqref{supp:eq-char-two-pairing} would have to equal $1$ (recall that
$-1=1$ in characteristic two).  After multiplication by
$z=\rho^{k+1}$, this condition says that one fixed quadratic polynomial
in $z$ vanishes for all $k=1,\ldots,q-2$.  The consecutive powers
$\rho^2,\ldots,\rho^{q-1}$ contain at least three distinct values when
$q\ge8$.  The quadratic is therefore identically zero, forcing
$\Norm(c_1)=\Norm(c_2)=0$, and hence $c_1=c_2=0$.  But then every nonzero
Fourier mode vanishes, a contradiction. \hfill$\square$

This establishes the characteristic restriction for the two-pole
kernel in Eq.~\eqref{supp:eq-char-two-kernel} when $q\ge8$.
Other kernels and characteristic-two AME constructions remain
separate possibilities.

\section{Finite-field representatives and exact checks}
\label{supp:sec-data}

{Each matrix is specified by its field, border entries,
and circulant kernel. The data below determine every entry through
Eq.~\eqref{supp:eq-border}, followed by the exact verification results.}

\subsection{Fields and universal representatives}
\label{supp:sec-representatives}

Write $q=p^e$.  In each row of Table~\ref{supp:tab-fields}, choose a root
$g$ of the displayed monic polynomial $m_g(X)$ over $\F_p$ and set
$\F_{q^2}=\F_p(g)$.  The polynomial is the minimal polynomial of $g$,
and $g$ has multiplicative order $q^2-1$; thus every nonzero field element
has a unique expression $g^j$ with $0\le j<q^2-1$.  Multiplication adds
exponents modulo $q^2-1$, while addition uses the relation $m_g(g)=0$.
These two rules completely specify the field arithmetic, including every
entry of the matrices below.  The symbols $0$ and $1=g^0$ have their usual
meaning.  Conjugation is $z\mapsto z^q$.

For the universal representatives, $\gamma=g^{q+1}$, $\Norm(b)=2$, and
$a=-b/{\bar b}$. Table~\ref{supp:tab-fields} supplies the field and border
data for the universal kernels listed below, through $q=27$.
The additional field needed for $M_{32,29}$ is specified separately
in Sec.~\ref{supp:sec-finite-certificates}.
\begin{table}[h!]
\caption{Primitive elements and their minimal polynomials over $\F_p$.
All displayed field entries are powers of the generator $g$ belonging to
the relevant row. The table covers all odd prime powers $q\le27$ and
specifies the universal matrices $A_q$ listed below.
These are explicit choices; the theorem holds for every admissible parameter tuple.}
\label{supp:tab-fields}
\small
\begin{ruledtabular}
\begin{tabular}{ccc r ccc}
$q$ & $p$ & $m_g(X)$ & $\operatorname{ord}(g)$ & $\gamma$ & $b$ & $a$\\
\colrule
3 & 3 & $X^{2}+X+2$ & 8 & $g^{4}$ & $g$ & $g^{2}$\\
5 & 5 & $X^{2}+X+2$ & 24 & $g^{6}$ & $g$ & $g^{8}$\\
7 & 7 & $X^{2}+5X+5$ & 48 & $g^{8}$ & $g^{40}$ & $g^{24}$\\
9 & 3 & $X^{4}+X+2$ & 80 & $g^{10}$ & $g^{44}$ & $g^{8}$\\
11 & 11 & $X^{2}+3X+6$ & 120 & $g^{12}$ & $g^{69}$ & $g^{90}$\\
13 & 13 & $X^{2}+9X+2$ & 168 & $g^{14}$ & $g$ & $g^{72}$\\
17 & 17 & $X^{2}+11X+6$ & 288 & $g^{18}$ & $g^{18}$ & $g^{144}$\\
19 & 19 & $X^{2}+15X+2$ & 360 & $g^{20}$ & $g$ & $g^{162}$\\
23 & 23 & $X^{2}+21X+19$ & 528 & $g^{24}$ & $g^{72}$ & $g^{264}$\\
25 & 5 & $X^{4}+2X^{3}+3X^{2}+3X+2$ & 624 & $g^{26}$ & $g^{318}$ & $g^{168}$\\
27 & 3 & $X^{6}+X+2$ & 728 & $g^{28}$ & $g^{13}$ & $g^{26}$\\
\end{tabular}
\end{ruledtabular}
\end{table}
For compact notation, define the finite-field tuple
\begin{equation}
 [k_1,\ldots,k_r]_g:=(g^{k_1},\ldots,g^{k_r}).
\end{equation}
For example, $[3,6]_g$ represents the field elements $(g^3,g^6)$.
Each kernel $\kappa=(\kappa_0,\ldots,\kappa_{q-2})$ is a vector of
length $q-1$. The following list specifies the universal
representatives for all odd prime powers $q\le27$.
{\small
\begingroup
\allowdisplaybreaks[1]
\begin{align}
q=3:\quad &[3,6]_g\nonumber\\
q=5:\quad &[23,8,4,21]_g\nonumber\\
q=7:\quad &[34,23,14,12,37,11]_g\nonumber\\
q=9:\quad &[33,18,39,26,64,17,5,52]_g\nonumber\\
q=11:\quad &[34,75,119,54,4,41,92,57,19,26]_g\nonumber\\
q=13:\quad &[145,122,81,107,82,100,92,21,139,87,90,102]_g\nonumber\\
q=17:\quad &[186,139,15,70,277,245,14,27,254,287,190,275,188,174,129,76]_g\nonumber\\
q=19:\quad &[44,297,75,7,169,303,173,216,62,191,30,119,185,108,78,72,254,126]_g\nonumber\\
q=23:\quad &[435,238,206,346,271,400,506,197,28,416,403,372,109,318,113,189,404,143,369,66,231,203]_g\nonumber\\
q=25:\quad &\bigl[32,489,259,436,55,274,97,328,186,381,295,395,220,509,542,564,244,387,219,86,169,336,527,496\bigr]_g\nonumber\\
q=27:\quad &\bigl[293,268,595,132,478,382,598,549,217,649,154,305,698,62,219,675,284,603,164,176,432,327,347,55,205,358\bigr]_g.\nonumber
\end{align}
\endgroup
}
Substitution in Eq.~\eqref{supp:eq-border} reconstructs each full matrix.
For $q=3$, this gives the matrix and state worked out in
Sec.~\ref{supp:sec-qutrit-example}.

\subsection{Nine explicit parent matrices}
\label{supp:sec-finite-certificates}

The nine explicit parents in Table~\ref{tab:finite-atlas} are
specified by matrices $M_{N,q}$ of order $m=N/2$ over $\F_{q^2}$.
Their generators are $G_{N,q}=[I_m\mid M_{N,q}]$.
Thus the matrix order $m$ determines the length $N=2m$, while $q$
determines the local dimension. The universal matrix $A_q$ corresponds
to $m=q$. All nine parents below lie in the range $N\le2q$.
Every field element below is written as a power of a primitive element
$g$. For $q\le27$, its minimal polynomial is given in
Table~\ref{supp:tab-fields}. For the additional matrix $M_{32,29}$, use
$\F_{29^2}=\F_{29}(g)$ with
\[
 g^2+21g+14=0,\qquad \operatorname{ord}(g)=29^2-1=840.
\]
The minimal polynomial fixes the field arithmetic, so the following
entries specify each matrix exactly. Four of the five parents with
$N<2q$ were obtained by searching the bordered-circulant form in
Eq.~\eqref{supp:eq-border}; $M_{20,13}$ uses an unbordered circulant.
The compact forms are
\begingroup
\allowdisplaybreaks[1]
\begin{align}
M_{16,9}&=B_{7}\bigl(1,1;[29,21,37,2,40,57,76]_g\bigr),
\label{supp:eq-cert-169}\\
M_{18,9}&=B_{8}\bigl(g^{64},g^{44};[68,37,3,35,6,79,30,72]_g\bigr),
\label{supp:eq-cert-189}\\
M_{22,11}&=B_{10}\bigl(g^{90},g^{69};[46,109,91,89,105,44,27,14,114,52]_g\bigr),
\label{supp:eq-cert-2211}\\
M_{20,13}&=\operatorname{Circ}_{10}\bigl([122,55,47,136,15,77,24,35,101,55]_g\bigr),
\label{supp:eq-cert-2013}\\
M_{26,13}&=B_{12}\bigl(g^{72},g;[77,26,142,1,106,33,160,31,108,23,125,46]_g\bigr),
\label{supp:eq-cert-2613}\\
M_{20,17}&=B_{9}\bigl(g^{227},g^{198};[7,68,31,126,119,239,147,106,114]_g\bigr),
\label{supp:eq-cert-2017}\\
M_{34,17}&=B_{16}\bigl(g^{144},g^{18};[186,139,15,70,277,245,14,27,254,287,190,275,188,174,129,76]_g\bigr),
\label{supp:eq-cert-3417}\\
M_{28,25}&=B_{13}\bigl(g^{312},1;[156,269,614,622,597,97,369,489,553,573,574,374,485]_g\bigr),
\label{supp:eq-cert-2825}\\
M_{32,29}&=B_{15}\bigl(g^{150},1;[120,51,445,446,32,493,181,370,650,209,17,88,334,305,639]_g\bigr).
\label{supp:eq-cert-3229}
\end{align}
\endgroup

The choice at $q=17$ agrees exactly with the universal formula:
$M_{34,17}=A_{17}$ with the parameters of Table~\ref{supp:tab-fields}.
The other three $N=2q$ matrices, $M_{18,9}$, $M_{22,11}$, and
$M_{26,13}$, are separately found representatives.
The circle family also realizes the parameters of $M_{20,17}$,
$M_{28,25}$, and $M_{32,29}$, while antipodal descent covers all nine
parameter pairs. The displayed matrices provide reproducible examples;
their monomial and local-Clifford equivalence to the general
representatives remains a separate question.

Exact finite-field arithmetic gives
\begin{equation}
 M_{N,q}\overline{M_{N,q}}^{\,\mathsf T}=-I_m
 \label{supp:eq-finite-gram}
\end{equation}
{and establishes superregularity with the minor checks
in Table~\ref{supp:tab-finite-audit}. The complementary orders for
$M_{34,17}$ follow from Eq.~\eqref{supp:eq-Jacobi}.
Checks of the universal $A_q$, including $A_{27}$, are reported
separately below.}
\begin{table}[t]
\caption{{Exact verification of the nine finite parents.
The counts include all nonempty square minors, except for
$M_{34,17}$, where $\dagger$ denotes orders $1,\ldots,8$.
Equation~\eqref{supp:eq-Jacobi} supplies its complementary orders.
All evaluated determinants were nonzero.}}
\label{supp:tab-finite-audit}
\begin{ruledtabular}
\begin{tabular}{ccccc}
$q$ & parent state & classical code & evaluated minors & projected state \\
\colrule
9  & $\AME(16,9)$ & $[16,8,9]_{81}$    & $12{,}869$       & $\AME(15,9)$ \\
9  & $\AME(18,9)$ & $[18,9,10]_{81}$   & $48{,}619$       & $\AME(17,9)$ \\
11 & $\AME(22,11)$ & $[22,11,12]_{121}$ & $705{,}431$      & $\AME(21,11)$ \\
13 & $\AME(20,13)$ & $[20,10,11]_{169}$ & $184{,}755$      & $\AME(19,13)$ \\
13 & $\AME(26,13)$ & $[26,13,14]_{169}$ & $10{,}400{,}599$ & $\AME(25,13)$ \\
17 & $\AME(20,17)$ & $[20,10,11]_{289}$ & $184{,}755$      & $\AME(19,17)$ \\
17 & $\AME(34,17)$ & $[34,17,18]_{289}$ & $1{,}166{,}803{,}109^{\dagger}$
                                              & $\AME(33,17)$ \\
25 & $\AME(28,25)$ & $[28,14,15]_{625}$ & $40{,}116{,}599$ & $\AME(27,25)$ \\
29 & $\AME(32,29)$ & $[32,16,17]_{841}$ & $601{,}080{,}389$& $\AME(31,29)$ \\
\end{tabular}
\end{ruledtabular}
\end{table}
{
The code and state parameters in the table follow from
Sec.~\ref{supp:sec-codes} with $B=M_{N,q}$ and $m=N/2$.

For explicit reconstruction, use the cyclic convention
$\operatorname{Circ}_{\mu}(\kappa)_{ij}=\kappa_{(j-i)\bmod\mu}$.
If $\kappa=[k_0,\ldots,k_{\mu-1}]_g$, the core entry is
$g^{k_{(j-i)\bmod\mu}}$; its border entries are $a,b$ as in
Eq.~\eqref{supp:eq-border}. This determines all expanded matrices,
including the unbordered $M_{20,13}$.
For example, the first two rows of $M_{16,9}$ are
\[
 (1,1,1,1,1,1,1,1),\qquad
 (1,g^{29},g^{21},g^{37},g^2,g^{40},g^{57},g^{76}).
\]
An exponent $0$ always denotes $g^0=1$; all displayed field entries
are nonzero. The compact formulas thus retain the full matrix data.
}

\subsection{Exact arithmetic checks}

{
Exact finite-field arithmetic checks the displayed representatives
independently of the analytic proof. For the universal matrices
through $q=13$, every square minor was evaluated and found nonzero.
The checks also covered all admissible $(\alpha,\tau)$ pairs:
\begin{center}
\begin{tabular}{crrrrrr}
$q$&3&5&7&9&11&13\\ \hline
Minors of $A_q$&19&251&3\,431&48\,619&705\,431&10\,400\,599\\
Admissible pairs&4&12&24&40&60&84
\end{tabular}
\end{center}
These counts are $\binom{2q}{q}-1$ and $(q^2-1)/2$, respectively.
It suffices to fix $\omega$ and $b$: changing the primitive element
$\omega$ permutes the finite row and column labels, while, for fixed
$\alpha,\tau,\omega$,
\begin{equation}
 A_q(b')=D_cA_q(b)D_c,\qquad
 b'=cb,\quad \Norm(c)=1,\quad D_c=\operatorname{diag}(c,1,\ldots,1).
 \label{supp:eq-b-change}
\end{equation}
Each minor is therefore preserved up to a nonzero factor.

For $A_{17}=M_{34,17}$, Hermitian orthogonality gives
$A_{17}^{-1}=-\overline{A_{17}}^{\,\mathsf T}$.
For equally sized index sets $I,J$, with complements $I^c,J^c$,
Jacobi's identity yields
\begin{equation}
 \det A_{17}[I,J]
 =\epsilon_{I,J}\det(A_{17})\det\!\bigl(A_{17}^{-1}[J^c,I^c]\bigr)
 =\epsilon'_{I,J}\det(A_{17})\,
   \overline{\det A_{17}[I^c,J^c]},
 \label{supp:eq-Jacobi}
\end{equation}
where $\epsilon_{I,J},\epsilon'_{I,J}\in\{1,-1\}$.
Thus orders $1,\ldots,8$ determine orders $9,\ldots,16$; the full
determinant is nonzero by Eq.~\eqref{supp:eq-gram}.
The exact counts were
\[
 \underbrace{197{,}602{,}305}_{\text{orders }1\text{--}6}
 +\underbrace{378{,}224{,}704}_{\text{order }7}
 +\underbrace{590{,}976{,}100}_{\text{order }8}
 =1{,}166{,}803{,}109.
\]
The order-$8$ traversal comprised two disjoint segments of
$320{,}892{,}000$ and $270{,}084{,}100$ minors.
Every evaluated determinant was nonzero, completing the verification
of $A_{17}$ and its $\AME(34,17)$ and $\AME(33,17)$ states.
For $q=25,27$, the checks covered all minors through order $4$,
with respective counts $165{,}403{,}125$ and $316{,}682{,}055$;
all were nonzero. These are limited-order checks of the universal
matrices, whose full superregularity follows from
Theorem~\ref{supp:thm-universal}.

We also computed the Schur square of the code generated by
$[I_q\mid A_q]$. For a linear code $\mathcal C\subseteq\F_{q^2}^n$,
\begin{equation}
 \mathcal C^{\star2}=
 \operatorname{span}_{\F_{q^2}}\{\bm c\star\bm c':
       \bm c,\bm c'\in\mathcal C\},\qquad
 (\bm c\star\bm c')_j=c_jc'_j.
 \label{supp:eq-schur-square}
\end{equation}
For an $[n,k]_{q^2}$ GRS code, including its extended form,
$\dim\mathcal C^{\star2}=\min(n,2k-1)$: products of evaluation
polynomials span degree at most $2k-2$, with squared coordinate
multipliers; the extended coordinate is the corresponding leading
coefficient. The computed dimensions at $q=17,25,27$ are $34,50,54$,
respectively, compared with $33,49,53$ for GRS codes.
They independently distinguish these three codes from GRS codes.
}

\section{Sources for the AME existence map}
\label{supp:sec-map-sources}

The existence map in the main text uses the archived Huber--Wyderka table
of 8 June 2026 as its historical baseline \cite{AMETable},
supplemented by standard constructions and subsequent results
\cite{Raissi2018,HuberGrassl2020,Shi2026,BevinsBidav2026}.
The literature baseline extends through
5 September 2026. The displayed update of 8 October also includes the
present theorem and the explicit $\AME(9,6)$, $\AME(9,10)$ and
$\AME(11,10)$ constructions
\cite{KitamuraAME96,AtlasAME910,KitamuraAME1110}.
The latter three cells are dark green, recording both the AME
constructions and the exclusion of stabilizer realizations under the
conventions specified below.
Their pinned sources are identified
as E1--E3 in the figures. Dots mark the odd-prime-power columns;
the full range $2\le N\le2q$ is filled, including cells absent from
the baseline. Figure~\ref{supp:fig-map-extended}
extends the display to $q=19,23,25,27,29$. Coordinates in the
following source discussion are $(d,N)$.

The entries $(13,16)$ and $(13,15)$ come from the $[[16,0,9]]_{13}$
code in Fig.~2 of Ref.~\cite{GrasslRotteler2015} and its one-party
projection, respectively.

The existence of $\AME(4,6)$ follows from the quantum solution of
Euler's officers problem \cite{Rather2022}; orthogonal Latin squares
supply $\AME(4,10)$ and $\AME(4,14)$
\cite{BoseShrikhandeParker1960}. The seven-party dark-green cells use
the existence theorem of Shi \emph{et al.}~\cite{Shi2026}.
These cells record both facts: an AME construction and a
nonexistence result for stabilizer states under the stated convention.

Red cells record nonexistence results for arbitrary pure states.
Their sources include the four- and seven-qubit nonexistence results
\cite{HiguchiSudbery2000,HuberGuhneSiewert2017}, Scott's general bounds
\cite{Scott2004}, and the shadow-inequality bounds
\cite{HuberShadow2018}, as recorded in the inherited table.

For $\AME(4,6)$, Ball and Simoens show that the 36 quantum officers
cannot all be product states \cite{BallSimoens2026}.
This is a different statement from the nonexistence of stabilizer
or graph-state realizations used in the overlays \cite{Cha2026,Wojcik2026}.
For composite $d$, the stabilizer exclusions refer to the generalized
$\mathbb Z_d$ Weyl--Heisenberg convention and its prime-power
decomposition as specified in Ref.~\cite{Cha2026}.
The canonical $\mathbb Z_d$ graph-state result of
Ref.~\cite{Wojcik2026} uses a different operator convention from
finite-field stabilizer states at $d=2^e$; its nonexistence statement
is applied only to the former.
Refs.~\cite{Danielsen2012} and \cite{BallMorenoSimoens2025}
supply the indicated nonexistence results for stabilizer states at
$(3,11)$ and $(4,8)$,
respectively.  Unmarked white cells indicate questions left unresolved in the cited
sources at the stated cutoff dates.

\begin{figure}[h]
\centering
\includegraphics[width=\textwidth,height=0.86\textheight,keepaspectratio]{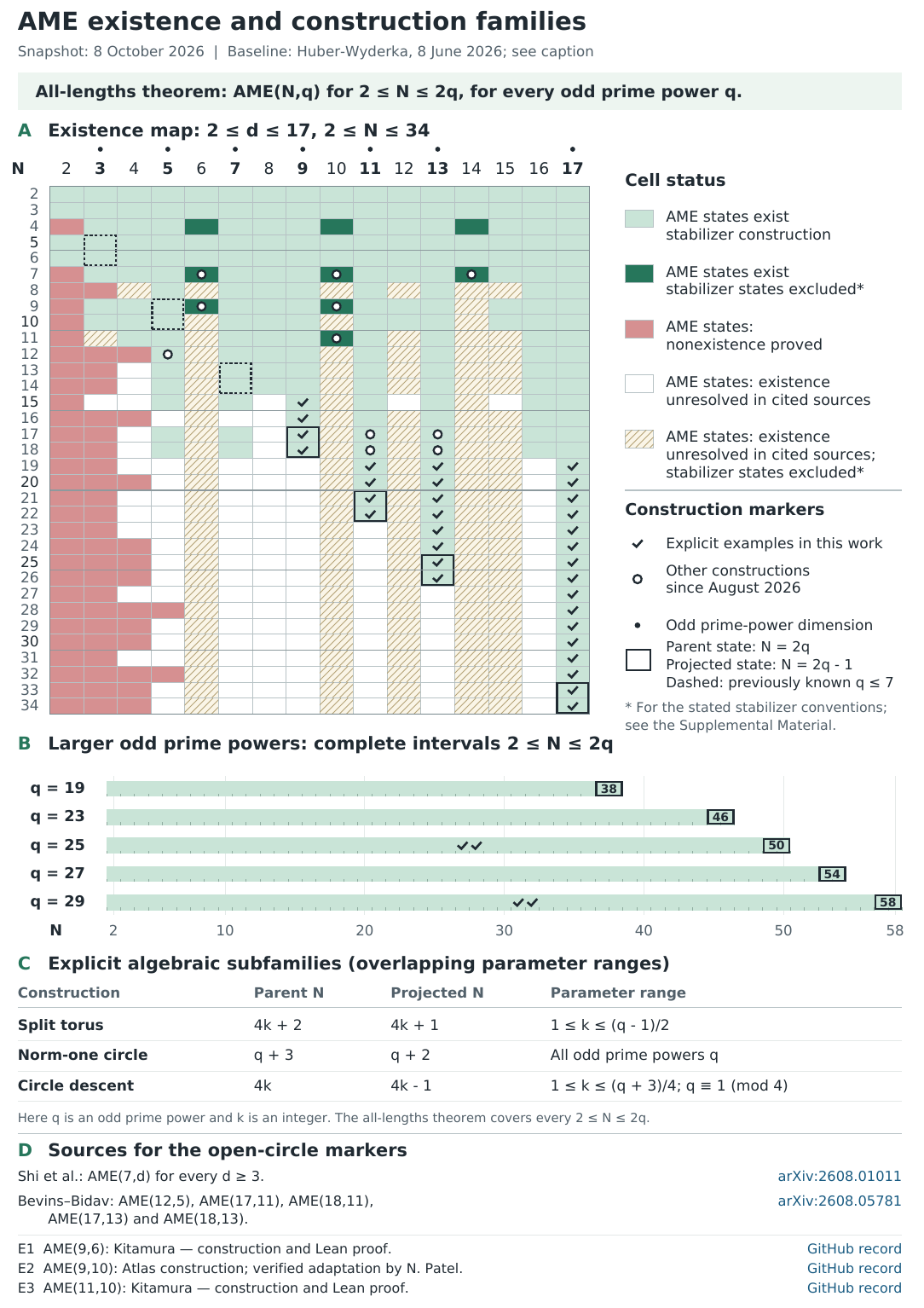}
\caption{{Extended AME existence map, based on the
Huber--Wyderka table \cite{AMETable} and the sources listed above.
Panel A shows $2\le d\le17$, $2\le N\le34$, with the status colors
and stabilizer conventions of Fig.~\ref{fig:ame-map}.
Panel B shows the full integer range $2\le N\le2q$ for
$q=19,23,25,27,29$, including $25=5^2$ and $27=3^3$.
Panel C summarizes the overlapping split, circle, and
circle-descent families. Panel D lists the corresponding
preprints and online construction records
\cite{Shi2026,BevinsBidav2026,KitamuraAME96,AtlasAME910,KitamuraAME1110}.}
Theorem~\ref{thm:main} covers every odd
prime power; the selected columns illustrate this general result.
Green cells and intervals assert existence, while checks identify
the explicit examples. Open circles mark other constructions since August 2026;
the markers carry no equivalence or priority claim.}
\label{supp:fig-map-extended}
\end{figure}

\end{document}